\documentclass[11pt]{article}

\usepackage[utf8]{inputenc}

\usepackage{amsmath,amsfonts,amsthm,amssymb,color}
\usepackage[usenames,dvipsnames,svgnames,table]{xcolor}
\definecolor{darkgreen}{rgb}{0.0,0,0.9}
\usepackage{mathtools}
\usepackage{authblk}
\usepackage{fullpage}
\usepackage{parskip}
\usepackage{comment}
\usepackage{tikz}
\usepackage{bbm}
\usepackage{dsfont}
\usepackage[sc]{mathpazo}
\usepackage[basic]{complexity}
\usepackage{algorithm2e}
\usepackage[colorlinks=true,citecolor=OliveGreen,linkcolor=BrickRed,urlcolor=BrickRed,pdfstartview=FitH]{hyperref}
\usepackage[capitalize,nameinlink]{cleveref}
\usepackage{tcolorbox}

\newtcolorbox{wbox}
{
	colback  = white,
}

\SetKwInOut{Input}{Input}
\SetKwInOut{Output}{Output}
\SetKwFunction{Uncross}{\textsc{Uncross}}
\SetKwFunction{MergeUncross}{\textsc{MergeUncross}}
\SetKwFunction{ConnectedComponents}{\textsc{ConnectedComponents}}
\SetKwFunction{PerfectMatching}{\textsc{PerfectMatching}}
\SetKwBlock{InParallel}{in parallel do}{end}
\SetKwFor{ParallelFor}{for}{in parallel do}{end}

\newcommand*{\suppress}[1]{}

\DeclareMathOperator{\rot}{rot}
\DeclareMathOperator{\pre}{pre}

\makeatletter
\def\thm@space@setup{%
	\thm@preskip= 10pt
	\thm@postskip=\thm@preskip 
}
\makeatother

\makeatletter
\renewcommand{\paragraph}{%
	\@startsection{paragraph}{4}%
	{\z@}{5pt}{-1em}%
	{\normalfont\normalsize\bfseries}%
}
\makeatother

\newtheorem{theorem}{Theorem}
\newtheorem{lemma}{Lemma}
\newtheorem{corollary}{Corollary}
\newtheorem{definition}{Definition}

\newtheorem{proposition}[theorem]{Proposition}
\newtheorem{claim}[theorem]{Claim}

\theoremstyle{definition}
\newtheorem{remark}[theorem]{Remark}

\newenvironment{fminipage}%
{\begin{Sbox}\begin{minipage}}%
		{\end{minipage}\end{Sbox}\fbox{\TheSbox}}

\def\union{\cup}
\def\intersect{\cap}
\def\Union{\bigcup}

\newcommand\yy{\boldsymbol{\mathit{y}}}

\newcommand{\Bc}{\mathcal{B}}
\newcommand{\Mc}{\mathcal{M}}
\newcommand{\Lc}{\mathcal{L}}                                                    
\newcommand{\rin}{\rho_{\text{in}}}
\newcommand{\rout}{\rho_{\text{out}}}

\newcommand{\st}{\text{s.t.}}

\newcommand{\Pbar}{\overline{P}}

\newcommand{\tail}{\text{tail}}
\newcommand{\head}{\text{head}}

\title{A Structural and Algorithmic Study of Stable Matching\\
Lattices of ``Nearby'' Instances, with Applications}

 \author[1]{Rohith Reddy Gangam}
 \author[2]{Tung Mai\footnote{This work was done while the author was a postdoctoral fellow at the University of California, Irvine.\\ 
 This work was supported in part by NSF grants CCF-1815901 and CCF-2230414.}}
 \author[1]{Nitya Raju}
 \author[1]{Vijay V.~Vazirani}

 \affil[1]{University of California, Irvine}
 \affil[2]{Adobe Research}

\date{}

\begin{document}
	\maketitle
	
	\begin{abstract}

Recently \cite{MV.robust} identified and initiated work on the new problem of understanding structural relationships between the lattices of solutions of two ``nearby'' instances of stable matching. They also gave an application of their work to finding a {\em robust stable matching}. However, the types of changes they allowed in going from instance $A$ to $B$ were very restricted, namely any one agent executes an {\em upward shift}.

In this paper, we allow any one agent to permute its preference list {\em arbitrarily}. Let $M_A$ and $M_B$ be the sets of stable matchings of the resulting pair of instances $A$ and $B$, and let $\Lc_A$ and $\Lc_B$ be the corresponding lattices of stable matchings. We prove that the matchings in $M_A \cap M_B$ form a sublattice of both $\Lc_A$ and $\Lc_B$ and those in $M_A \setminus M_B$ form a join semi-sublattice of $\Lc_A$. These properties enable us to obtain a polynomial time algorithm for not only finding a stable matching in $M_A \cap M_B$, but also for obtaining  the partial order, as promised by Birkhoff's Representation Theorem \cite{Birkhoff}, thereby enabling us to generate all matchings in this sublattice. 

Our algorithm also helps solve a version of the robust stable matching problem. We discuss another potential application, namely obtaining new insights into the incentive compatibility properties of the Gale-Shapley Deferred Acceptance Algorithm.
\end{abstract}

\section{Introduction}
\label{sec.intro}

The seminal 1962 paper of Gale and Shapley \cite{GaleS} introduced the stable matching problem and gave the Deferred Acceptance (DA) Algorithm for it. In the process, they initiated the field of matching-based market design. Over the years, numerous researchers unearthed the remarkably deep and pristine structural properties of this problem -- this led to polynomial time algorithms for for a host of problems, in particular those addressing various operations related to the lattice of stable matchings, see details below as well as in the books \cite{Knuth-book, GusfieldI, Manlove-book, Roth-Sotomayor, MM.book}. 

Recently \cite{MV.robust} identified and initiated work on a new problem which appears to be fundamental and deserving of an in-depth study, namely understanding structural relationships between the lattices of solutions of two ``nearby'' instances. \cite{MV.robust} had given an application of their work to finding a {\em robust stable matching} as described below. Such pairs of instances arise naturally in an even more important context: the study of incentive compatibility of the DA Algorithm. Let $A$ and $B$ be the given instance and the one in which one of the agents manipulates its preference list in order to get a better match. The types of manipulations allowed in \cite{MV.robust} were very restricted, namely any one agent executes an {\em upward shifts}, see definition below. They left the open problem of tackling more general changes. 

In this paper, we allow any one agent to permute its preference list {\em arbitrarily}. Let $A$ and $B$ be the resulting pair of instances, let $M_A$ and $M_B$ be the sets of their stable matchings and $\Lc_A$ and $\Lc_B$ be the corresponding lattices of stable matchings. We prove that the matchings in $M_A \cap M_B$ form a sublattice of both $\Lc_A$ and $\Lc_B$ and those in $M_A \setminus M_B$ form a join semi-sublattice, see definitions in Section \ref{sec.overview}. This enables is to obtain a polynomial time algorithm for not only finding a stable matching in $M_A \cap M_B$, but also to obtain the partial order, promised by Birkhoff's Representation Theorem \cite{Birkhoff}, which helps generate all matchings in this sublattice. We also apply our algorithm to a more general setting for robust stable matching than the one given in \cite{MV.robust}.

The setting defined in \cite{MV.robust} was the following: Let $A$ be an instance of stable matching on $n$ workers and $n$ firms. A {\em domain of errors}, $D$, is defined via an operation called {\em upward shift}: For a firm $f$, assume its preference list in instance $A$ is $\{ \ldots,w_1,w_2, \ldots, w_k, w , \ldots \}$. Move up the position of worker $w$ so $f$'s list becomes $\{ \ldots, w, w_1,w_2, \ldots, w_k , \ldots \}$. An analogous operation is defined on a worker $w$'s list; again some firm $f$ on its list is moved up. For each firm and each worker, consider all possible shifts to get the domain D; clearly, $|D| = \binom{2n}{1} \binom{n}{2} =$ $O(n^3)$. Assume that {\em one error} is chosen from $D$ via a given discrete probability distribution over $D$ to obtain instance $B$. A {\em robust stable matching} is a matching that is stable for $A$ and maximizes the probability of being stable for $B$. A polynomial time algorithm was given for finding such a matching.

Since we allow an {\em arbitrary permutation} to be applied to any {\em one worker or any one firm's} preference list, our domain of errors, say $T$, has size $2n (n!)$. Let $S \subseteq T$ and define a {\em fully robust stable matching w.r.t. $S$} to be a matching that is stable for $A$ and for {\em each} of the $|S|$ instances obtained by introducing one error from $S$. We give an $O(|S| p(n))$ algorithm to determine if such a matching exists and if so to find one, where $p$ is a polynomial function. In particular, if $S$ is polynomial sized, then our algorithm runs in polynomial time. Clearly, this notion is weaker than the previous one, since we cannot extend it to the probabilistic setting; we leave that as an open problem, see Section \ref{sec.discussion}. 

There is a simple modification of the Deferred Acceptance algorithm (Algorithm~\ref{alg:daalgorithm}) given in Appendix~\ref{app:algorithms} that works when errors are on one side. However extending this algorithm to errors on both side results in an algorithm (Algorithm~\ref{alg:daalgorithm2}) that has exponential runtime. This motivates the study and characterization of sublattices in the lattice of stable matchings. 

Conway, see \cite{Knuth-book}, proved that the set of stable matchings of an instance forms a finite distributive lattice; see definitions in Section \ref{sec.lattice}. Knuth \cite{Knuth-book} asked if every finite distributive lattice is isomorphic to the lattice arising from an instance of stable matching. A positive answer was provided by Blair \cite{Blair}; for a much better proof, see \cite{GusfieldI}. A key fact about such lattices is Birkhoff's Representation Theorem \cite{Birkhoff}, which has also been called {\em the fundamental theorem for finite distributive lattices}, e.g., see \cite{Stanley}. It states that corresponding to such a lattice, $\Lc$, there is a partial order, say $\Pi$, such that $\Lc$ is isomorphic to $L(\Pi)$, the lattice of closed sets of $\Pi$ (see Section \ref{sec.lattice} for details). We will say that $\Pi$ {\em generates} $\Lc$. 

The following important question arose in the design of our algorithm: For a specified
sublattice $\Lc'$ of $\Lc$, obtain partial order $\Pi'$ from $\Pi$ such that $\Pi'$ generates $\Lc'$. Our answer to this question requires a study Birkhoff’s Theorem from this angle; we are not aware of any previous application of Birkhoff’s Theorem in this manner. We define a set of operations
called compressions; when a compression is applied to a partial order $\Pi$, it yields a partial
order $\Pi'$ on (weakly) fewer elements. The following implication of Birkhoff’s Theorem is
useful for our purposes:

\begin{theorem}
	\label{thm:generalization} 
	There is a one-to-one correspondence between the compressions of $\Pi$ and the sublattices of 
	$L(\Pi)$ such that if sublattice $\Lc'$ of $L(\Pi)$ corresponds to compression $\Pi'$, 
	then $\Lc'$ is generated by $\Pi'$.
\end{theorem}

The proof for Theorem~\ref{thm:generalization} using stable matching lattices is given in Section~\ref{gen:proof} for completeness.

In the case of stable matchings, $\Pi$ can be defined using the notion of {\em rotations}; see Section \ref{sec.lattice} for a formal definition. Since the total number of rotations of a stable matching instance is at most $O(n^2)$, $\Pi$ has a succinct description even though $\Lc$ may be exponentially large. Our main algorithmic result is:

\begin{theorem}
	\label{thm:main}
There is an algorithm for checking if there is a fully robust stable matching w.r.t. any set $S \subseteq T$ in time  $O(|S| p(n))$, where $p$ is a polynomial function. Moreover, if the answer is yes, the set of all such matchings forms a sublattice of $\Lc$ and our algorithm finds a partial order that generates it.
\end{theorem}

The importance of the stable matching problem lies not only in its efficient computability but also its good incentive compatibility properties. Roth \cite{Roth-IC-1982economics} showed there is no stable matching procedure for which truthful revelation of preferences is a dominant strategy for all agents on both sides. On the other hand, Roth \cite{Roth-IC-1982economics} and Dubins and Freedman \cite{Dubins-1981} independently proved that the DA Algorithm is {\em dominant-strategy incentive compatible (DSIC)} for the proposing side. Additionally, Dubins and Freedman \cite{Dubins-1981} proved that no coalition of agents on the proposing side can simultaneously improve all of their matches by altering their preferences, provided all agents outside this set report their preferences truthfully. 

This opened up the use of the DA Algorithm in a host of highly consequential applications and led to the award of the 2012 Nobel Prize in Economics to Roth and Shapley. Among the applications was matching students to public schools in big cities, such as NYC and Boston, see \cite{Atila-Sonmez-school, NYC-school, Boston}. In this application, the proposing side is taken to be the students; clearly, their best strategy is to report preference lists truthfully and not waste time and effort on ``gaming'' the system. In Section \ref{sec.discussion} we give a hypothetical situation regarding incentive compatibility in which Theorem \ref{thm:main} plays a role.


\subsection{Related work}
\label{sec.related}

The two topics, of stable matching and the design of algorithms that produce solutions that are robust to errors,  have been studied extensively for decades and there are today several books on each of them, e.g., see \cite{Knuth-book, GusfieldI, Manlove-book} and \cite{CaE:06, BEN:09}. Yet, there is a paucity of results at the intersection of these two topics. Indeed, before the publication of \cite{MV.robust}, we are aware of only two previous works \cite{aziz1,aziz2}. We remark that the notion of robustness studied in \cite{MV.robust} was quite different from that of the previous two works as detailed below.

Aziz et al. \cite{aziz1} considered the problem of finding stable matching under uncertain linear preferences. They proposed three different uncertainty models:
\begin{enumerate}
\item Lottery Model: Each agent has a probability distribution
over strict preference lists, independent of other agents. 
\item Compact Indifference Model: Each agent has a single weak preference
list in which ties may exist. All linear order extensions of this
weak order have equal probability.
\item Joint Probability Model: A probability distribution over preference profiles
is specified.
\end{enumerate}
They showed that finding the matching with highest probability of being stable is NP-hard for the Compact Indifference Model and the Joint Probability Model. For the very special case that preference lists of one side are certain and the number of uncertain agents of the other side are bounded by a constant, they gave a polynomial time algorithm that works for all three models.

The joint probability model is the most powerful and closest to our setting. The main difference is that in their model, there is no base instance, which is called $A$ in our model. The opportunity of finding new structural results arises from our model precisely because we need to consider two ``nearby'' instances, namely $A$ and $B$ as described above.

Aziz et al. \cite{aziz2} introduced a pairwise probability model in which
each agent gives the probability of preferring one agent over another for all possible pairs.
They showed that the problem of finding a matching with highest probability of being stable is NP-hard even when no agent has a cycle in its certain preferences (i.e., the ones that hold with probability 1).

\subsubsection{A matter of nomenclature}
\label{sec.nomen}

Assigning correct nomenclature to a new issue under investigation is clearly critical for ease of 
comprehension. In this context we wish to mention that very recently,
Genc et al. \cite{genc2} defined the notion of an $(a, b)$-supermatch as follows: this 
is a stable matching in which if any $a$ pairs break
up, then it is possible to match them all off by changing the partners
of at most $b$ other pairs, so the resulting matching is also stable.
They showed that it is NP-hard to decide if there is an
$(a, b)$-supermatch. They also gave a polynomial time algorithm for a very restricted version
of this problem, namely given a stable matching and a number $b$, decide if it is a 
$(1, b)$-supermatch. Observe that since the given instance may have exponentially many stable 
matchings, this does not yield a polynomial time algorithm even for deciding
if there is a stable matching which is a $(1, b)$-supermatch for a given $b$. 

Genc et al. \cite{genc1} also went on to defining the notion of the most robust stable matching, namely 
a $(1, b)$-supermatch where $b$ is minimum. We would like to point out that ``robust'' is a misnomer
in this situation and that the name ``fault-tolerant'' is more appropriate.
In the literature, the latter is used to
describe a system which continues to operate even in the event of failures and the former is 
used to describe a system which is able to cope with erroneous inputs, e.g., 
see the following pages from Wikipedia \cite{Robust, FT}.

\subsection{Overview of structural and algorithmic ideas}
\label{sec.overview}

We start by giving a short overview of the structural facts proven in \cite{MV.robust}. 
Let $A$ and $B$ be two instances of stable matching over $n$ workers and $n$ firms, with sets of 
stable matchings $\Mc_A$ and $\Mc_B$, and lattices $\Lc_A$ and $\Lc_B$, respectively.
Let $\Pi$ be the poset on rotations such that $L(\Pi) = \Lc_A$; in particular, for a closed set $S$, let $M(S)$ denote the stable matching corresponding to $S$. It is easy to see that if $B$ is obtained from $A$ by changing (upshifts only) the lists of only one side, either workers or firms, but not both, then the matchings in $\Mc_A \cap \Mc_B$ form a sublattice of each of the two lattices (Proposition ~\ref{prop.sublattice}). Furthermore, if $B$ is obtained by applying a shift operation, then $\Mc_{A \setminus B} = \Mc_A \setminus \Mc_B$ is also a sublattice of $\Lc_A$. Additionally, there is at most one rotation, $\rin$, that leads from  $\Mc_{A} \cap \Mc_{B}$ to $\Mc_{A \setminus B}$ and at most one rotation, $\rout$, that leads from $\Mc_{A \setminus B}$ to $\Mc_{A} \cap \Mc_{B}$; moreover, these rotations can be efficiently found. Finally, for a closed set $S$ of $\Pi$, $M(S)$ is stable for instance $B$ iff $\rin \in S \ \Rightarrow \ \rout \in S$.

With a view to extending the results of \cite{MV.robust}, we consider the following abstract  question. Suppose instance $B$ is such that $\Mc_A \cap \Mc_B$ and $\Mc_{A \setminus B}$ are both sublattices of $\Lc_A$, i.e., $\Mc_A$ is partitioned into two sublattices. Then, is there a polynomial time algorithm for finding a matching in $\Mc_A \cap \Mc_B$? Our answer to this question is built on the following structural fact: There exists a sequence of rotations $r_0, r_1, \ldots , r_{2k}, r_{2k+1}$ such that a closed set of $\Pi$ generates a matching in $\Mc_A \cap \Mc_B$ iff it contains $r_{2i}$ but not $r_{2i+1}$ for some $0 \leq i \leq k$ (Proposition \ref{prop.sub}). Furthermore, this sequence of rotations can be found in polynomial time (see Section \ref{sec.sublattice}). Our generalization of Birkhoff's Theorem described in the Introduction is an important ingredient in this algorithm. At this point, we do not know of any concrete error pattern, beyond shift, for which this abstract setting applies.

Next, we address the case that $\Mc_{A \setminus B}$ is not a sublattice of $\Lc_A$. We start by proving that if $B$ is obtained by permuting the preference list of any one worker, then $\Mc_{A \setminus B}$ must be a join semi-sublattice of $\Lc_A$ (Lemma \ref{lem:MABsemi}); an analogous statement holds if the preference list of any one firm is permuted. Hence we study a second abstract question, namely lattice $\Lc_A$ is partitioned into a sublattice and a join semi-sublattice (see Section \ref{sec.semi}). These two abstract questions are called {\bf Setting I and Setting II}, respectively, in this paper.

For Setting II, we characterize a compression that yields a partial order $\Pi'$, such that $\Pi'$ generates the sublattice consisting of matchings in $\Mc_{A} \cap \Mc_{B}$ (Theorem \ref{thm:semi}). We also characterize closed sets of $\Pi$ such that the corresponding matchings lie in this sublattice; however, the characterization is too elaborate to summarize succinctly (see Proposition \ref{prop.gen}). Edges forming the required compression can be found  efficiently (Theorem \ref{thm:algFindFlower}), hence leading to an efficient algorithm for finding a matching in $\Mc_{A} \cap \Mc_{B}$.

Finally, consider the setting given in the Introduction, with $T$ being the super-exponential set of
all possible errors that can be introduced in instance $A$ and $S \subset T$.
We show that the set of all matchings that are stable for $A$ and for each of the instances obtained by introducing one error from $S$ forms a sublattice of $\Lc$ and we obtain a compression of $\Pi$ that generates this sublattice (Section \ref{sec.fullyRobust}). 
Each matching in this sublattice is a fully robust stable matching. Furthermore, given a weight function on all worker-firm pairs,
we can obtain, using the algorithm of \cite{MV.weight}, a maximum (or minimum) weight fully robust stable matching.

	\section{Preliminaries}

\subsection{The stable matching problem and the lattice of stable matchings}
\label{subsection.latticeOfSM}

The stable matching problem takes as input a set of workers $\mathcal{W} = \{w_1, w_2, \ldots , w_n\}$ and a set of firms $\mathcal{F} = \{f_1, f_2, \ldots , f_n\}$; each agent has a complete preference ranking over the set of the other side. The notation $w_i <_f w_j$ indicates that firm $f$ strictly prefers $w_j$ to $w_i$ in its preference list. Similarly, $f_i <_w f_j$ indicates that the worker $w$ strictly prefers $f_j$ to $f_i$ in its list.

A matching $M$ is a one-to-one correspondence between $\cal W$ and $\cal F$. For each pair $(w, f) \in M$, $w$ is called the partner of $f$ in $M$ (or $M$-partner) and vice versa. 
For a matching $M$, a pair $(w, f) \not \in M$ is said to be \emph{blocking} if 
they prefer each other to their partners. A matching $M$ is \emph{stable} if there is no blocking pair for $M$.

Let $M$ and $M'$ be two stable matchings. We say that $M$ \emph{dominates} $M'$, denoted by 
$M \preceq M'$, if every worker weakly prefers its partner in $M$ to $M'$. Define the relation 
{\em predecessor} as the transitive closure of dominates. For two stable matchings, $M_1$ and $M_2$,   stable matching $M$ is a {\em common predecessor} of $M_1$ and $M_2$ if it is a predecessor of both $M_1$ and $M_2$. Furthermore, $M$ is a {\em lowest common predecessor} of $M_1$ and $M_2$ if it is a common predecessor $M_1$ and $M_2$, and if $M'$ is another common predecessor, then $M'$ cannot be a predecessor of $M$. Analogously, one can define the notions of {\em successor}  and {\em highest common successor} (definitions are omitted). This dominance partial order has the following property: For any two stable $M_1$ and $M_2$, their lowest common predecessor is unique and their highest common successor is unique, i.e., the partial order is a {\em lattice}; the former is called the {\em meet}, denoted $M_1 \wedge M_2$, and the latter is called the {\em join}, denoted $M_1 \vee M_2$. One can show that $M_1 \wedge M_2$, is the matching that results when each worker chooses its more preferred partner from $M_1$ and $M_2$; it is easy to show that this matching is also stable. Similarly, $M_1 \vee M_2$ is the matching that results when each worker chooses its less preferred partner from $M_1$ and $M_2$; this matching is also stable. 

These operations distribute, i.e.,
given three stable matchings $M, M', M''$,
$$ M \vee (M' \wedge M'') =  (M \vee M') \wedge (M \vee M'') \ \ \mbox{and} \ \
M \wedge (M' \vee M'') = (M \wedge M') \vee (M \wedge M'').$$

It is easy to see that the lattice must contain a matching, $M_0$, that dominates all others
and a matching $M_z$ that is dominated by all others.
$M_0$ is called the \emph{worker-optimal matching}, since in it, each worker is matched to its most
preferred firm among all stable matchings. This is also the {\em firm-pessimal matching}.
Similarly, $M_z$ is the {\em worker-pessimal} or \emph{firm-optimal matching}.


\subsection{Birkhoff's Theorem and rotations}
\label{sec.lattice}

It is easy to see that the family of closed sets (also called lower sets, Definition~\ref{closedset}) of a partial order, say $\Pi$, is closed under 
union and intersection and forms a distributive lattice, with join and meet being these two 
operations, respectively; let us denote it by $L(\Pi)$.
Birkhoff's theorem \cite{Birkhoff}, which has also been called {\em the fundamental theorem for 
finite distributive lattices}, e.g., see \cite{Stanley}, states that corresponding to any 
finite distributed lattice, $\Lc$, there is a partial order, say $\Pi$, whose lattice of closed sets 
$L(\Pi)$ is isomorphic to $\Lc$, i.e., $\Lc \cong L(\Pi)$. 
We will say that $\Pi$ {\em generates} $\Lc$. 

For the lattice of stable matchings, the partial order $\Pi$ defined in Birkhoff's Theorem, has
additional useful structural properties. First, its elements are rotations. 
A {\em rotation} takes $r$ matched worker-firm pairs in a fixed order, say 
$\{w_0f_0, w_1f_1,\ldots, w_{r-1}f_{r-1}\}$, and ``cyclically'' changes the matches of these $2r$ 
agents. The number $r$, the $r$ pairs, and the order among the pairs are so chosen that when a rotation is
applied to a stable matching containing all $r$ pairs, the resulting matching is also stable;
moreover, there is no valid rotation on any subset of these $r$ pairs, under
any ordering. Hence, a rotation can be viewed as a minimal
change to the current matching that results in a stable matching.

 Any worker--firm pair, $(w, f)$,
belongs to at most one rotation. Consequently, the set $R$ of rotations underlying $\Pi$ 
satisfies $|R|$ is $O(n^2)$, and hence, $\Pi$ is a succinct representation of $\Lc$; the
latter can be exponentially large. $\Pi$ will be called the {\em rotation poset} for $\Lc$.

Second, the rotation poset helps traverse the lattice as follows.
For any closed set $S$ of $\Pi$, the corresponding stable matching $M(S)$ can be obtained
as follows: start from the worker-optimal matching in the lattice and apply
the rotations in set $S$, in any topological order consistent with $\Pi$. The resulting
matching will be $M(S)$. In particular, applying all rotations in $R$, starting from the worker-optimal
matching, leads to the firm-optimal matching.

The following process yields a rotation for a stable matching $M$. For a worker $w$ let $s_M(w)$ denote the first firm $f$ on $w$'s list such that $f$ strictly prefers $w$ to its $M$-partner. Let $next_M(w)$ denote the partner in $M$ of firm $s_M(f)$. A \emph{rotation} $\rho$ \emph{exposed} in $M$ is an ordered list of pairs $\{w_0f_0, w_1f_1,\ldots, w_{r-1}f_{r-1}\}$ such that for each $i$, $0 \leq i \leq r-1$, $w_{i+1}$ is $next_M(w_i)$, where $i+1$ is taken modulo $r$. In this paper, we assume that the subscript is taken modulo $r$ whenever we mention a rotation. Notice that a rotation is cyclic and the sequence of pairs can be rotated. $M / \rho$ is defined to be a matching in which each worker not in a pair of $\rho$ stays matched to the same firm and each worker $w_i$ in $\rho$ is matched to $f_{i+1} = s_M(w_i)$. It can be proven that $M / \rho$ is also a stable matching. The transformation from $M$ to $M / \rho$ is called the \emph{elimination} of $\rho$ from $M$.

\begin{lemma}[\cite{GusfieldI}, Theorem 2.5.4]
	\label{lem:seqElimination}
	Every rotation appears exactly once in any sequence of elimination from $M_0$ to $M_z$.
\end{lemma}

Let $\rho = \{w_0f_0, w_1f_1,\ldots, w_{r-1}f_{r-1}\}$ be a rotation. For $0 \leq i \leq r-1$, we say that $\rho$ \emph{moves $w_i$ from $f_i$ to $f_{i+1}$}, and \emph{moves $f_i$ from $w_{i}$ to $w_{i-1}$}. If $f$ is either $f_i$ or is strictly between $f_{i}$ and $f_{i+1}$ in $w_i$'s list, then we say that $\rho$ \emph{moves $w_i$ below $f$}. Similarly, $\rho$ \emph{moves $f_i$ above} $w$ if $w$ is $w_i$ or between $w_i$ and $w_{i-1}$ in $f_i$'s list.

\subsection{The rotation poset}

A rotation $\rho'$ is said to \emph{precede} another rotation $\rho$, denoted by $\rho' \prec \rho$, if $\rho'$ is eliminated in every sequence of eliminations from $M_0$ to a stable matching in which $\rho$ is exposed. 
If $\rho'$ precedes $\rho$, we also say that $\rho$ \emph{succeeds} $\rho'$.
If neither $\rho' \prec \rho$ nor $\rho' \succ \rho$, we say that $\rho'$ and $\rho$ are 
\emph{incomparable}
Thus, the set of rotations forms a partial order via
this precedence relationship. The partial order on rotations is called \emph{rotation poset} and denoted by $\Pi$.

\begin{lemma}[\cite{GusfieldI}, Lemma 3.2.1]
	\label{lem:pre2}
	For any worker $w$ and firm $f$, there is at most one rotation that moves $w$ to $f$, $w$ below $f$, or $f$ above $w$. Moreover, if $\rho_1$ moves $w$ to $f$ and $\rho_2$ moves $w$ from $f$ then $\rho_1 \prec \rho_2$.
\end{lemma}

\begin{lemma}[\cite{GusfieldI}, Lemma 3.3.2]
	\label{lem:computePoset}
	$\Pi$ contains at most $O(n^2)$ rotations and can be computed in polynomial time.
\end{lemma}

Consequently, $\Pi$ is a succinct representation of $\Lc$; the
latter can be exponentially large.

\begin{definition}
	\label{closedset}
A closed set of a poset is a set $S$ of elements of the poset such that if an element is in 
$S$ then all of its predecessors are also in $S$.
\end{definition}

There is a one-to-one relationship between the stable matchings and the closed subsets of $\Pi$. Given a closed set $S$, the correponding matching $M$ is found by eliminating the rotations starting from $M_0$ according to the topological ordering of the elements in the set $S$. 
We say that $S$ \emph{generates} $M$ and that $\Pi$ {\em generates the lattice} $\Lc$ of all
stable matchings of this instance.

Let $S$ be a subset of the elements of a poset, and let $v$ be an element in $S$. 
We say that $v$ is a \emph{minimal} element in $S$ if there is no predecessors of $v$ in $S$. Similarly, $v$ is a \emph{maximal} element in $S$ if it has no successors in $S$.

The \emph{Hasse diagram} of a poset is a directed graph with a vertex for each element in poset,
and an edge from $x$ to $y$ if $x \prec y$ and there is no $z$ such that $x \prec z \prec y$.
In other words, all precedences implied by transitivity are suppressed.

\subsection{Sublattice and semi-sublattice}
A \emph{sublattice} $\Lc'$ of a distributive lattice $\Lc$ is subset of $\Lc$ such that 
for any two elements $x,y \in \Lc$, $x \vee y \in \Lc'$ and  $x \wedge y \in \Lc'$ whenever $x,y \in \Lc'$, where $\vee$ and $\wedge$ are the join and meet operations of lattice $\Lc$.

A \emph{join semi-sublattice} $\Lc'$ of a distributive lattice $\Lc$ is subset of $\Lc$ such that for any two elements $x,y \in \Lc$, $x \vee y \in \Lc'$ whenever $x,y \in \Lc'$. 

Similarly, \emph{meet semi-sublattice} $\Lc'$ of a distributive lattice $\Lc$ is subset of $\Lc$ such that for any two elements $x,y \in \Lc$, $x \wedge y \in \Lc'$ whenever $x,y \in \Lc'$. 

Note that $\Lc'$ is a sublattice of $\Lc$ iff $\Lc'$ is both join and meet semi-sublattice of $\Lc$.

\begin{proposition} \label{prop.sublattice}
	Let $A$ be an instance of stable matching and let $B$ be another instance obtained from $A$ by changing the lists of only one side, either workers or firms, but not both. Then the matchings in $\Mc_A \cap \Mc_B$ form a sublattice in each of the two lattices. 
\end{proposition}
\begin{proof}
	It suffices to show that $\Mc_A \cap \Mc_B$ is a sublattice of $\Lc_A$. Assume $| \Mc_A \cap \Mc_B| > 1$ and let $M_1$ and $M_2$ be two different matchings in $\Mc_A \cap \Mc_B$. 
	Let $\vee_A$ and $\vee_B$ be the join operations under $A$ and $B$ respectively. Likewise, let $\wedge_A$ and $\wedge_B$ be the meet operations under $A$ and $B$.
	
	By definition of join operation in Section~\ref{subsection.latticeOfSM}, $M_1 \vee_A M_2$ is the matching obtained by assigning each worker to its less preferred partner (or equivalently, each firm to its more preferred partner) from $M_1$ and $M_2$ according to instance $A$.
	Without loss of generality, assume that $B$ is an instance obtained from $A$ by changing the lists of only firms. Since the list of each worker is identical in $A$ and $B$, its less preferred partner 
	from $M_1$ and $M_2$ is also the same in $A$ and $B$. 
	Therefore, $M_1 \vee_A M_2 =  M_1 \vee_B M_2 $. A similar argument can be applied to show that $M_1 \wedge_A M_2 =  M_1 \wedge_B M_2$.
	
	Hence, $M_1 \vee_A M_2$ and $M_1 \wedge_A M_2$ are both in $\Mc_A \cap \Mc_B$ as desired.
\end{proof}

\begin{corollary}
	\label{cor.sublatticeIntersection}
	Let $A$ be an instance of stable matching and let $B_1, \ldots ,B_k$ be other instances obtained from $A$ each by changing the lists of only one side, either workers or firms, but not both. Then the matchings in $\Mc_A \cap \Mc_{B_1} \cap \ldots \cap \Mc_{B_k}$ form a sublattice in $\Mc_A$. 
\end{corollary}

\begin{proof}
	Assume $| \Mc_A \cap \Mc_{B_1} \cap \ldots \cap \Mc_{B_k}| > 1$ and let $M_1$ and $M_2$ be two different matchings in $\Mc_A \cap \Mc_{B_1} \cap \ldots \cap \Mc_{B_k}$. Therefore, 
	$M_1$ and $M_2$ are in $\Mc_A \cap \Mc_{B_i}$ for each $1 \leq i \leq k$. By Proposition~\ref{prop.sublattice}, $\Mc_A \cap \Mc_{B_i}$ is a sublattice of $\Lc_A$. Hence, $M_1 \vee_A M_2$ and $M_1 \wedge_A M_2$ are in $\Mc_A \cap \Mc_{B_i}$ for each $1 \leq i \leq k$. 
	The claim then follows. 
\end{proof}

This corollary serves as motivation for the algorithm (Algorithm~\ref{alg:daalgorithm}) given in Appendix~\ref{app:algorithms}. This modified Deferred Algorithm works when errors are only on one side. Extending this algorithm to errors on both sides results in an algorithm (Algorithm~\ref{alg:daalgorithm2}) that has an exponential run time.

This motivates us to characterize sublattices in the lattice of stable matchings. In Section~\ref{subsection.semisublatticeNecessarySufficient}, 
we show that for any instance $B$ obtained by permuting the preference list of one worker or one firm,
$\Mc_{A \setminus B}$ forms a semi-sublattice of $\Lc_A$ (Lemma~\ref{lem:MABsemi}).  
In particular, if the list of a worker is permuted, $\Mc_{A \setminus B}$ forms a join semi-sublattice of $\Lc_A$, and if the list of a firm is permuted, $\Mc_{A \setminus B}$ forms a meet semi-sublattice of $\Lc_A$. 
In both cases, $\Mc_A \intersect \Mc_B$ is a sublattice of $\Lc_A$ and of $\Lc_B$ as shown in Proposition~\ref{prop.sublattice}.
	
	\section{Birkhoff's Theorem on Sublattices}
\label{sec.generalization}

Let $\Pi$ be a finite poset. For simplicity of notation, in this paper we will assume that $\Pi$ must
have {\em two dummy elements} $s$ and $t$; the remaining elements will be called {\em proper
elements} and the term {\em element} will refer to proper as well as dummy elements. The element
$s$ precedes all other elements and $t$ succeeds all other elements in $\Pi$. A {\em proper
closed set} of $\Pi$ is any closed set that contains $s$ and does not contain $t$.
It is easy to see that the set of all proper closed sets of $\Pi$ form a distributive 
lattice under the operations of set intersection and union. We will denote this lattice by 
$L(\Pi)$. The following has also been called {\em the fundamental theorem for finite distributive lattices}.

\begin{theorem}
(Birkhoff \cite{Birkhoff})	Every finite distributive lattice $\Lc$ is isomorphic to $L(\Pi)$, 
for some finite poset $\Pi$.
\end{theorem}


Our application of Birkhoff's Theorem deals with the sublattices of a finite distributive 
lattice. First, in Definition \ref{def:compression} we state the critical operation of 
\emph{compression of a poset}.  

\begin{definition}
	\label{def:compression}
Given a finite poset $\Pi$, first partition its elements; each subset will be called
a \emph{meta-element}. Define the following precedence relations among the meta-elements: 
if $x,y$ are elements of $\Pi$ such that $x$ is in meta-element $X$, $y$ is in meta-element $Y$ 
and $x$ precedes $y$, then $X$ precedes $Y$. Assume that these precedence relations yield a
partial order, say $Q$, on the meta-elements (if not, this particular partition is not useful 
for our purpose). Let $\Pi'$ be any partial order on the meta-elements such that the precedence
relations of $Q$ are a subset of the precedence relations of $\Pi'$. Then $\Pi'$ will be called
a {\em compression} of $\Pi$. Let $A_s$ and $A_t$ denote the meta-elements of $\Pi'$ containing $s$ 
and $t$, respectively.
\end{definition}

\begin{figure}
	\begin{wbox}
		\begin{minipage}[c]{0.5\textwidth}
			\centering
			\def\svgscale{0.4}
			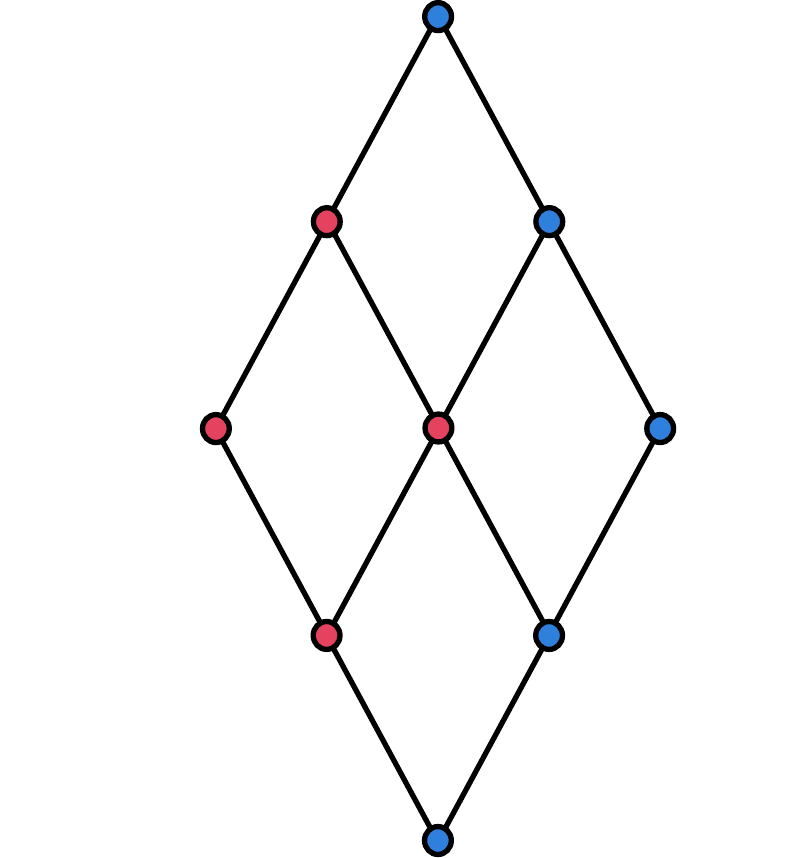
			
			~~~~$\Lc$
		\end{minipage}
		\begin{minipage}[c]{0.5\textwidth}
			\centering
			\def\svgscale{0.4}
			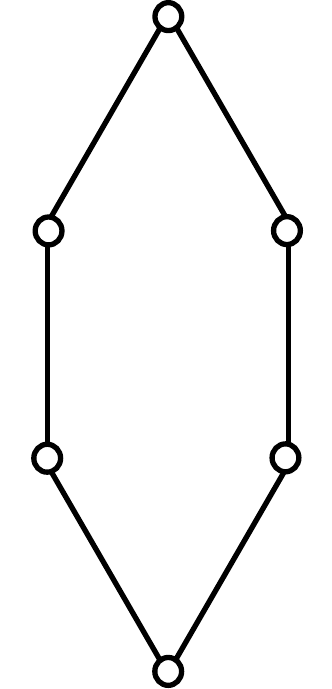
			
			$P$	
		\end{minipage}
		\par  
		\begin{minipage}[c]{0.55\textwidth}
			\centering
			\def\svgscale{0.4}
			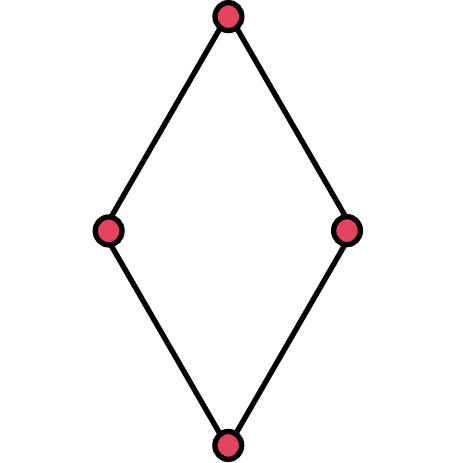
			
			$P_1$	
		\end{minipage}
		\begin{minipage}[c]{0.45\textwidth}
			\centering
			\def\svgscale{0.4}
			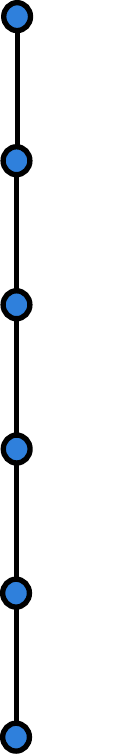
			
			$P_2$	
		\end{minipage}
	\end{wbox}
	\caption{Two examples of compressions. Lattice $\Lc = L(P)$.
$P_1$ and $P_2$ are compressions of $P$, and they generate the sublattices in $\Lc$, of red 
and blue elements, respectively.
	}
	\label{ex:compression} 
\end{figure} 

For examples of compressions see Figure~\ref{ex:compression}.
Clearly, $A_s$ precedes all other meta-elements in $\Pi'$ and $A_t$ succeeds all other meta-elements 
in $\Pi'$. Once again, by a {\em proper closed set of $\Pi'$} we mean a closed set of $\Pi'$ that
contains $A_s$ and does not contain $A_t$. Then the lattice formed by the set of all proper closed 
sets of $\Pi'$ will be denoted by $L(\Pi')$.

\subsection{Proof of Birkhoff's Theorem using Stable Matching Lattices}
\label{gen:proof}

We will prove Theorem \ref{thm:generalization} in the context of stable matching lattices; this is
w.l.o.g. since stable matching lattices are as general as finite distributive lattices.
In this context, the proper elements of partial order $\Pi$ will be rotations, 
and meta-elements are called \emph{meta-rotations}.
Let $\Lc = L(\Pi)$ be the corresponding stable matching lattice. 

Clearly it suffices to show that:
\begin{itemize}
	\item Given a compression $\Pi'$, $L(\Pi')$ is isomorphic to a sublattice of $\Lc$.
	\item Any sublattice $\Lc'$ is $L(\Pi')$ for some compression $\Pi'$.
\end{itemize}
These two proofs are given
in Sections \ref{sec.forward} and \ref{sec.backward}, respectively.

\subsubsection{$L(\Pi')$ is isomorphic to a sublattice of $L(\Pi)$}
\label{sec.forward}

Let $I$ be a closed subset of $\Pi'$; clearly $I$ is a set of meta-rotations. Define $\rot(I)$ to
be the union of all meta-rotations in $I$, i.e.,
\[ \rot(I) = \{ \rho \in A: A \text{ is a meta-rotation in } I \} .\]

We will define the process of {\em elimination of a meta-rotation} $A$ of $\Pi'$ to be the
elimination of the rotations in $A$ in an order consistent with partial order $\Pi$. Furthermore,
{\em elimination of meta-rotations in $I$} will mean starting from stable matching $M_0$ in
lattice $\Lc$ and eliminating all meta-rotations in $I$
in an order consistent with $\Pi'$. Observe that this is equivalent to starting from stable matching 
$M_0$ in $\Lc$ and eliminating all rotations
in $\rot(I)$ in an order consistent with partial order $\Pi$. This follows from
Definition~\ref{def:compression}, since
if there exist rotations $x,y$ in $\Pi$ such that $x$ is in meta-rotation $X$, $y$ is in meta-rotation $Y$ and $x$ precedes $y$, then $X$ must also precede $Y$. 
Hence, if the elimination of all rotations in $\rot(I)$ gives matching $M_I$, then elimination
of all meta-rotations in $I$ will also give the same matching.

Finally, to prove the statement in the title of this section, it suffices to observe that if
$I$ and $J$ are two proper closed sets of the partial order $\Pi'$ then
\[ \rot(I \cup J) = \rot(I) \cup \rot(J) \ \ \ \ \mbox{and} \ \ \ \ 
\rot(I \cap J) = \rot(I) \cap \rot(J) . \]
It follows that the set of matchings obtained by elimination of meta-rotations in a proper
closed set of $\Pi'$ are closed under the operations of meet and join and hence form a 
sublattice of $\Lc$.

\subsubsection{Sublattice $\Lc'$ is generated by a compression $\Pi'$ of $\Pi$}
\label{sec.backward}

We will obtain compression $\Pi'$ of $\Pi$ in stages. First, we show how to partition the set of 
rotations of $\Pi$ to obtain the meta-rotations of $\Pi'$. We then find precedence relations among 
these meta-rotations to obtain $\Pi'$. Finally, we show $L(\Pi') = \Lc'$.

Notice that $\Lc$ can be represented by its Hasse diagram $H(\Lc)$. 
Each edge of $H(\Lc)$
{\em contains} exactly one (not necessarily unique) rotation of $\Pi$.
Then, by Lemma~\ref{lem:seqElimination}, for any two stable matchings $M_1, M_2 \in \Lc$ such that
$M_1 \prec M_2$, all paths from $M_1$ to $M_2$ in $H(\Lc)$ contain the same set of rotations.

\begin{definition} 
	For $M_1, M_2 \in \Lc'$, $M_2$ is said to be an {\em $\Lc'$-direct successor} of $M_1$ iff 
	$M_1 \prec M_2$ and there is no $M \in \Lc'$ such that $M_1 \prec M \prec M_2$.
	Let $M_1 \prec \ldots \prec M_k$ be a sequence of matchings in $\Lc'$ such that $M_{i+1}$ 
	is an $\Lc'$-direct successor of $M_{i}$ for all $1 \leq i \leq k-1$. Then any
	path in $H(\Lc)$ from $M_1$ to $M_k$ containing $M_{i}$, for all $1 \leq i \leq k-1$, is called 
	an {\em $\Lc'$-path}. 
\end{definition}

Let $M_{0'}$ and $M_{z'}$ denote the worker-optimal and firm-optimal matchings, respectively, 
in $\Lc'$. For $M_1, M_2 \in \Lc'$ with $M_1 \prec M_2$, let $S_{M_1, M_2}$ denote the set of 
rotations contained on any $\Lc'$-path from $M_1$ to $M_2$. Further,
let $S_{M_0,M_{0'}}$ and $S_{M_{z'},M_z}$ denote the set of rotations contained on any path from 
$M_0$ to $M_{0'}$ and $M_{z'}$ to $M_z$, respectively in $H(\Lc)$. Define the following set
whose elements are sets of rotations.
$$ \mathcal{S} =  \{S_{M_i,M_j} \ | \  M_j \text{ is an $\Lc'$-direct successor of } M_i, \ 
\text{for every pair of matchings} \ M_i, M_j \ \text{in} \ \Lc' \} \Union $$
$$ \ \  \{ S_{M_0,M_{0'}}, \ S_{M_{z'},M_z} \}. $$

\begin{lemma} 
	\label{lem:partition}
	$\mathcal{S}$ is a partition of $\Pi$.
\end{lemma}
\begin{proof}
	First, we show that any rotation must be in an element of $\mathcal{S}$. Consider a path $p$ from $M_0$ to $M_z$ in the $H(\Lc)$ such that $p$ goes from $M_{0'}$ to $M_{z'}$ via an $\Lc'$-path. Since $p$ is a path from $M_0$ to $M_z$, all rotations of $\Pi$ are contained on
	$p$ by Lemma~\ref{lem:seqElimination}. Hence, they all appear in the sets in $\mathcal{S}$.
	
	Next assume that there are two pairs $(M_1,M_2) \not = (M_3,M_4)$ of $\Lc'$-direct successors 
	such that $S_{M_1,M_2} \not = S_{M_3,M_4}$ and $X = S_{M_1,M_2} \intersect S_{M_3,M_4} \not = \emptyset$. The set of rotations eliminated from $M_0$ to $M_2$ is  
	\[S_{M_0,M_2} = S_{M_0,M_1} \union S_{M_1,M_2}.\]
	Similarly, 
	\[S_{M_0,M_4} = S_{M_0,M_3} \union S_{M_3,M_4}.\]
	Therefore, 
	\[S_{M_0,M_2 \vee M_3} =S_{M_0,M_3} \union S_{M_1,M_2} \union S_{M_0,M_1}. \]
	\[S_{M_0,M_1 \vee M_4} = S_{M_0,M_3} \union S_{M_3,M_4} \union S_{M_0,M_1}. \]
	Let $M = (M_2 \vee M_3) \wedge (M_1 \vee M_4)$, we have  
	\[S_{M_0, M} = S_{M_0,M_3} \union S_{M_0,M_1} \union X.\]
	Hence, 
	\[S_{M_0, M \wedge M_2} = S_{M_0,M_1} \union X.\]
	Since $X \subset S_{M_1,M_2}$ and $S_{M_1,M_2} \intersect S_{M_0,M_1} = \emptyset$, $X \intersect S_{M_0,M_1} = \emptyset$. Therefore, 
	\[S_{M_0,M_1} \subset S_{M_0, M \wedge M_2} \subset S_{M_0,M_2},\] 
	and hence $M_2$ is not a $\Lc'$-direct successor of $M_1$, leading to a contradiction.
\end{proof}

We will denote $S_{M_0,M_{0'}}$ and $S_{M_{z'},M_{z}}$ by $A_s$ and $A_t$, respectively.
The elements of $\mathcal{S}$ will be the meta-rotations of $\Pi'$.  Next, we need to define  
precedence relations among these meta-rotations to complete the construction of $\Pi'$. 
For a meta-rotation $A \in \mathcal{S}$, $A \neq A_t$, define the following subset of $\Lc'$:
\[ \Mc^A = \{M \in \Lc' \ \mbox{such that} \ A \subseteq S_{M_0,M} \} .\]

\begin{lemma}
	\label{lem:latticeContainingMetaRotaion}
	For each meta-rotation $A \in \mathcal{S}$, $A \neq A_t$, $\Mc^A$ forms a sublattice $\Lc^A$ of $\Lc'$.
\end{lemma}

\begin{proof}
	Take two matchings $M_1,M_2$ such that $S_{M_0,M_1}$ and $S_{M_0,M_2}$ are supersets of $A$. Then $S_{M_0,M_1 \wedge M_2} = S_{M_0,M_1} \intersect S_{M_0,M_2}$ and $S_{M_0,M_1 \vee M_2} = S_{M_0,M_1} \union S_{M_0,M_2}$ are also supersets of $A$.
\end{proof}

Let $M^A$ be the worker-optimal matching in the lattice $\Lc^A$. Let $p$ be any $\Lc'$-path from 
$M_{0'}$ to $M^A$ and let $\pre(A)$ be the set of meta-rotations appearing before $A$ on $p$. 

\begin{lemma}
	The set $\pre(A)$ does not depend on $p$. Furthermore, on any $\Lc'$-path from $M_{0'}$ containing 
	$A$, each meta-rotation in $\pre(A)$ appears before $A$.
	\label{lem:preMetaPoset} 
\end{lemma}

\begin{proof}
	Since all paths from $M_{0'}$ to $M^A$ give the same set of rotations, 
	all $\Lc'$-paths from $M_{0'}$ to $M^A$ give the same set of meta-rotations.
	Moreover, $A$ must appear last in the any $\Lc'$-path from $M_{0'}$ to $M^A$;
	otherwise, there exists a matching in $\Lc^A$ preceding $M^A$, giving a contradiction.
	It follows that $\pre(A)$ does not depend on $p$.
	
	Let $q$ be an $\Lc'$-path from $M_{0'}$ that contains matchings $M', M \in \Lc'$, where $M$ is 
	an $\Lc'$-direct successor of $M'$. Let $A$ denote the meta-rotation that is contained on edge 
	$(M', M)$. Suppose there is a meta-rotation $A' \in \pre(A)$ such that $A'$ does not appear before
	$A$ on $q$. Then $S_{M_0,M^A \wedge M} = S_{M_0,M^A} \intersect S_{M_0,M}$ contains $A$ but 
	not $A'$. Therefore $M^A \wedge M$ is a matching in $\Lc^A$ preceding $M^A$, giving is a 
	contradiction.	Hence all matchings in $\pre(A)$ must appear before $A$ on all such paths $q$.
\end{proof}

Finally, add precedence relations from all meta-rotations in $\pre(A)$ to $A$, for each
meta-rotation in $\mathcal{S} - \{A_t\}$. Also, add precedence relations from all meta-rotations 
in $\mathcal{S} - \{A_t\}$ to $A_t$. This completes the construction of $\Pi'$. Below we show
that $\Pi'$ is indeed a compression of $\Pi$, but first we need to establish that this construction 
does yield a valid poset.

\begin{lemma}
	$\Pi'$ satisfies transitivity and anti-symmetry.
\end{lemma}
\begin{proof}
	First we prove that $\Pi'$ satifies transitivity.
	Let $A_1, A_2, A_3$ be meta-rotations such that $A_1 \prec A_2$ and $A_2 \prec A_3$.
	We may assume that $A_3 \not = A_t$.
	Then $A_1 \in \pre(A_2)$ and $A_2 \in \pre(A_3)$. 
	Since $A_1 \in \pre(A_2)$, $S_{M_0,M^{A_2}}$ is a superset of $A_1$.
	By Lemma~\ref{lem:latticeContainingMetaRotaion}, $M^{A_1} \prec M^{A_2}$. 
	Similarly, $M^{A_2} \prec M^{A_3}$. Therefore $M^{A_1} \prec M^{A_3}$, and hence $A_1 \in \pre(A_3)$.
	
	Next we prove that $\Pi'$ satisfies anti-symmetry. Assume that there exist meta-rotations $A_1, A_2$ such that
	$A_1 \prec A_2$ and $A_2 \prec A_1$. Clearly $A_1, A_2 \not = A_t$. Since $A_1 \prec A_2$, $A_1 \in \pre(A_2)$. Therefore, $S_{M_0,M^{A_2}}$ is a superset of $A_1$. It follows that $M^{A_1} \prec M^{A_2}$. Applying a similar argument we get $M^{A_2} \prec M^{A_1}$. Now, we get a contradiction,
	since $A_1$ and $A_2$ are different meta-rotations.
\end{proof}

\begin{lemma}
	$\Pi'$ is a compression of $\Pi$.
\end{lemma}
\begin{proof}
	Let $x,y$ be rotations in $\Pi$ such that $x \prec y$. 
	Let $X$ be the meta-rotation containing $x$ and $Y$ be the meta-rotation containing $y$.
	It suffices to show that $X \in \pre(Y)$.
	Let $p$ be an $\Lc'$-path from $M_0$ to $M^Y$.
	Since $x \prec y$, $x$ must appear before $y$ in $p$.
	Hence, $X$ also appears before $Y$ in $p$.
	By Lemma~\ref{lem:preMetaPoset}, $X \in \pre(Y)$ as desired.
\end{proof}

Finally, the next two lemmas prove that $L(\Pi') = \Lc'$.

\begin{lemma}
	\label{lem:LPsubsetLbar}
	Any matching in $L(\Pi')$ must be in $\Lc'$.
\end{lemma}
\begin{proof}
	For any proper closed subset $I$ in $\Pi'$, let $M_I$ be the matching generated by eliminating meta-rotations in $I$.	 
	Let $J$ be another proper closed subset in $\Pi'$ such that $J = I \setminus \{A\}$, 
	where $A$ is a maximal meta-rotation in $I$. 
	Then $M_J$ is a matching in $\Lc'$ by induction. 
	Since $I$ contains $A$, $S_{M_0,M_I} \supset A$. Therefore, $M^A \prec M_I$. 
	It follows that $M_I = M_J \vee M^A \in \Lc'$.
\end{proof}

\begin{lemma} Any matching in $\Lc'$ must be in $L(\Pi')$.
\end{lemma}
\begin{proof}
	Suppose there exists a matching $M$ in $\Lc'$ such that $M \not \in L(\Pi')$. 
	Then it must be the case that $S_{M_0,M}$ cannot be partitioned into meta-rotations which
	form a closed subset of $\Pi$. Now there are two cases.
	
	First, suppose that $S_{M_0,M}$ can be partitioned into meta-rotations, but they do not form a closed subset of $\Pi'$.
	Let $A$ be a meta-rotation such that $S_{M_0,M} \supset A$, and there exists $B \prec A$ such that $S_{M_0,M} \not \supset B$.
	By Lemma~\ref{lem:latticeContainingMetaRotaion}, $M \succ M^A$ and hence $S_{M_0,M}$ is a superset
	of all meta-rotations in $\pre(A)$, giving is a contradiction. 
	
	Next, suppose that $S_{M_0,M}$ cannot be partitioned into meta-rotations in $\Pi'$. Since the set of meta-rotations partitions $\Pi$, there exists a meta-rotation $X$ such that $Y = X \intersect S_{M_0,M}$ is a non-empty subset of $X$.
	Let $J$ be the set of meta-rotations preceding $X$ in $\Pi$. 
	
	$(M_J \vee M) \wedge M^X$ is the matching generated by 
	meta-rotations in $J \union Y$.
	Obviously, $J$ is a closed subset in $\Pi'$. Therefore, $M_J \in L(\Pi')$.
	By Lemma~\ref{lem:LPsubsetLbar}, $M_J \in \Lc'$.
	Since $M,M^X \in \Lc'$, $(M_J \vee M) \wedge M^X \in \Lc'$ as well. 
	The set of rotations contained on a path from $M_J$ to $(M_J \vee M) \wedge M^X$ in $H(\Lc)$ is exactly $Y$.
	Therefore, $Y$ can not be a subset of any meta-rotation, contradicting the fact that $Y = X \intersect S_{M_0,M}$ is a non-empty subset of $X$.
\end{proof}


	\subsection{An alternative view of compression}
\label{sec.alternative}

In this section we give an alternative definition of compression of a poset; this will be used
in the rest of the paper. The advantage of this definition is that it is much 
easier to work with for the applications presented later. Its drawback
is that several different sets of edges may yield the same compression. Therefore, this definition is not suitable for stating a one-to-one correspondence between sublattices of $\Lc$ and compressions of $\Pi$. Finally we show that any compression $\Pi'$ obtained using the first definition can also be obtained via the second definition and vice versa (Proposition \ref{prop.eq}), hence showing that the two definitions are equivalent for our purposes.

We are given a poset $\Pi$ for a stable matching instance; let
$\Lc$ be the lattice it generates. Let $H(\Pi)$ denote the Hasse diagram of $\Pi$. 
Consider the following operations to derive a new poset $\Pi'$: Choose a set $E$ of directed edges
to add to $H(\Pi)$ and let $H_E$ be the resulting graph. Let $H'$ be the graph obtained by 
shrinking the strongly connected components of $H_E$; each strongly connected component will be
a meta-rotation of $\Pi'$. The edges which are not shrunk will
define a DAG, $H'$, on the strongly connected components. These edges give precedence relations
among meta-rotation for poset $\Pi'$.

Let $\Lc'$ be the sublattice of $\Lc$ generated by $\Pi'$. We will say that the
set of edges $E$ \emph{defines} $\Lc'$. 
It can be seen that each set $E$ uniquely defines a sublattice $L(\Pi')$; however,
there may be multiple sets that define the same sublattice. 
See Figure~\ref{ex:edgeSets} for examples of sets of edges which define sublattices.

\begin{proposition}
\label{prop.eq}
	The two definitions of compression of a poset are equivalent.
\end{proposition}
\begin{proof}
	Let $\Pi'$ be a compression of $\Pi$ obtained using the first definition. Clearly, for each
meta-rotation in $\Pi'$, we can add edges to $\Pi$ so the strongly connected component created
is precisely this meta-rotation. Any additional precedence relations introduced among incomparable
meta-rotations can also be introduced by adding appropriate edges.

The other direction is even simpler, since each strongly connected component can be defined to 
be a meta-rotation and extra edges added can also be simulated by introducing new precedence
constraints. 
\end{proof}

\begin{figure}
	\begin{wbox}
		\begin{minipage}[c]{0.33\textwidth}
			\centering
			\def\svgscale{0.4}
			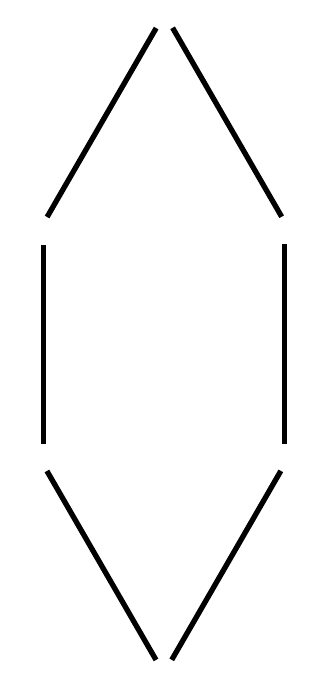
			
			$E_1$
		\end{minipage}
		\begin{minipage}[c]{0.32\textwidth}
			\centering
			\def\svgscale{0.4}
			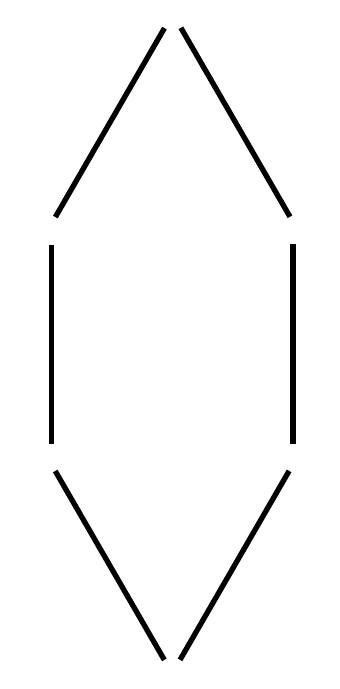
			
			$E_2$
		\end{minipage}
		\begin{minipage}[c]{0.3\textwidth}
			\centering
			\def\svgscale{0.4}
			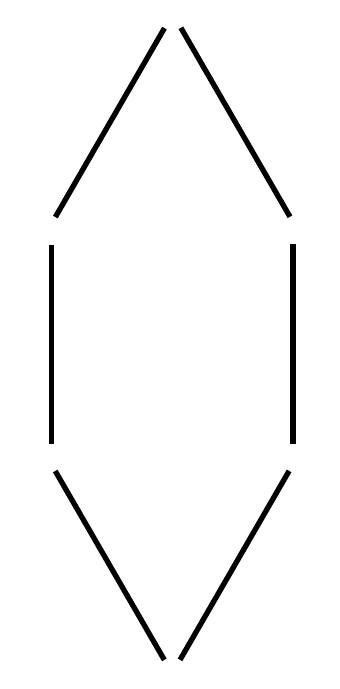
			
			$E_3$
		\end{minipage}
	\end{wbox}
	\caption{$E_1$ (red edges) and $E_2$ (blue edges) define the sublattices in 
Figure~\ref{ex:compression}, of red and blue elements, respectively. $E_2$ and $E_3$ define the same compression and represent the same sublattice. All black edges in $E_1, E_2$ and $E_3$ are directed from top to bottom (not shown in the figure). 	}
	\label{ex:edgeSets} 
\end{figure} 

For a (directed) edge $e = uv \in E$, $u$ is called the \emph{tail} and $v$ is 
called the \emph{head} of $e$.
Let $I$ be a closed set of $\Pi$. Then we say that:
\begin{itemize}
	\item $I$ \emph{separates} an edge $uv \in E$ if $v \in I$ and $u \not \in I$. 
	\item $I$ \emph{crosses} an edge $uv \in E$ if $u \in I$ and $v \not \in I$. 
\end{itemize}
If $I$ does not separate or cross any edge $uv \in E$, $I$ is called a \emph{splitting} set 
w.r.t. $E$.

\begin{lemma}
	\label{lem:separating}
	Let $\Lc'$ be a sublattice of $\Lc$ and $E$ be a set of edges defining $\Lc'$. 
	A matching $M$ is in $\Lc'$ iff the closed subset $I$ generating $M$ does not separate any edge $uv \in E$.
\end{lemma}  

\begin{proof}
	Let $\Pi'$ be a compression corresponding to $\Lc'$.
	By Theorem~\ref{thm:generalization}, the matchings in $\Lc'$ are generated by eliminating rotations in closed subsets of $\Pi'$. 
	
	First, assume $I$ separates $uv \in E$. Moreover, assume $M \in \Lc'$ for the sake of contradiction, and let $I'$ be the closed subset of $\Pi'$ corresponding to $M$. Let $U$ and $V$ be the meta-rotations containing $u$ and $v$ respectively.
	Notice that the sets of rotations in $I$ and $I'$ are identical. Therefore, $V \in I'$ and $U \not \in I'$. 
	Since $uv \in E$, there is an edge from $U$ to $V$ in $H'$. Hence, $I'$ is not a closed subset of $\Pi'$.
	
	Next, assume that $I$ does not separate any $uv \in E$. We show that the rotations in $I$ can be partitioned into meta-rotations in a closed subset $I'$ of $\Pi'$. If $I$ cannot be partitioned into meta-rotations, there must exist a meta-rotation $A$ such that $A \intersect I$ is a non-empty proper subset of $A$. Since $A$ consists of rotations in a strongly connected component of $H_E$, there must be an edge $uv$ from $A \setminus I$ to $A \intersect I$ in $H_E$. Hence, $I$ separates $uv$. Since $I$ is a closed subset, $uv$ can not be an edge in $H$. Therefore, $uv \in E$, which is a contradiction. It remains to show that the set of meta-rotations partitioning $I$ is a closed subset of $\Pi'$. Assume otherwise, there exist meta-rotation $U \in I'$ and $V \not \in I'$ such that there exists an edge from $U$ to $V$ in $H'$. Therefore, there exists $u \in U$, $v \in V$ and $uv \in E$, which is a contradiction.
\end{proof}

\begin{remark}
	\label{remark:addedEdge}
We may assume w.l.o.g. that the set $E$ defining $\Lc'$ is {\em minimal} in the following sense:
There is no edge $uv \in E$ such that $uv$ is not separated by any closed set of $\Pi$. Observe that
if there is such an edge, then $E \setminus \{uv\}$ defines the same sublattice $\Lc'$. 
Similarly, there is no edge $uv \in E$ such that each closed set separating $uv$ also separates another edge in $E$.
\end{remark}

\begin{definition}
\label{def.poset}
W.r.t. an element $v$ in a poset $\Pi$, we define four useful subsets of $\Pi$:
\begin{align*}
I_v &= \{r \in \Pi: r \prec v \} \\
J_v & = \{r \in \Pi: r \preceq v\} = I_v \union \{v\} \\
I'_v &= \{r \in \Pi: r \succ v \} \\
J'_v & = \{r \in \Pi: r \succeq v\} = I'_v \union \{v\}
\end{align*} 
Notice that $I_v, J_v, \Pi \setminus I'_v, \Pi \setminus J'_v$ are all closed sets. 
\end{definition}

\begin{lemma}
	\label{lem:prime}
	Both $J_v$ and $\Pi \setminus J'_u$ separate $uv$ for each $uv \in E$.
\end{lemma}
\begin{proof}
Since $uv$ is in $E$, $u$ cannot be in $J_v$; otherwise, there is no closed subset separating 
$uv$, contradicting Remark~\ref{remark:addedEdge}. Hence, $J_v$ separates $uv$ for all $uv$ in $E$.
	
	Similarly, since $uv$ is in $E$, $v$ cannot be in $J'_u$. Therefore, $\Pi \setminus J'_u$ contains $v$ but not $u$, and thus separates $uv$.
\end{proof}

\section{Setting I} 
\label{sec.sublattice}

Under Setting I, the given lattice $\Lc$ has sublattices $\Lc_1$ and $\Lc_2$ such that 
$\Lc_1$ and $\Lc_2$ partition $\Lc$. The main structural fact for this setting is:

\begin{theorem}
	\label{thm:alternating}
	Let $\Lc_1$ and $\Lc_2$ be sublattices of $\Lc$ such that $\Lc_1$ and $\Lc_2$ partition $\Lc$. 
	Then there 
	exist sets of edges $E_1$ and $E_2$ defining $\Lc_1$ and $\Lc_2$ such that they
	form an alternating path from $t$ to $s$. 
\end{theorem}

We will prove this theorem in the context of stable matchings. 
Let $E_1$ and $E_2$ be any two sets of edges defining $\Lc_1$ and $\Lc_2$,
respectively. 
We will show that $E_1$ and $E_2$ can be adjusted so that they form an alternating path from $t$ to $s$,
without changing the corresponding compressions. 

\begin{lemma}
	\label{lem:path}
	There must exist a path from $t$ to $s$ composed of edges in $E_1$ and $E_2$.
\end{lemma}

\begin{proof}
	Let $R$ denote the set of vertices reachable from $t$ by a path of edges in $E_1$ and $E_2$. Assume by contradiction that $R$ does not contain $s$. Consider the matching $M$ generated by rotations in $\Pi \setminus R$. Without loss of generality, assume that $M \in \Lc_1$. 
	By Lemma~\ref{lem:separating}, $\Pi \setminus R$ separates an edge $uv \in E_2$. Therefore, $u \in R$ and $v \in \Pi \setminus R$. Since $uv \in E_2$, $v$ is also reachable from $t$ by a path of edges in $E_1$ and $E_2$.
\end{proof}

Let $Q$ be a path from $t$ to $s$ according to Lemma~\ref{lem:path}. 
Partition $Q$ into subpaths $Q_1, \ldots, Q_k$ such that each $Q_i$ consists of edges in either $E_1$ or $E_2$ and 
$E(Q_i) \intersect E(Q_{i+1}) = \emptyset$ for all $1 \leq i \leq k-1$.
Let $r_i$ be the rotation at the end of $Q_i$ except for $i = 0$ where $r_0 = t$. Specifically, $t = r_0 \rightarrow r_1 \rightarrow \ldots \rightarrow r_k = s$ in $Q$.
We will show that each $Q_i$ can be replaced by a direct edge from $r_{i-1}$ to $r_i$, and furthermore,
all edges not in $Q$ can be removed.

\begin{lemma}
	\label{lem:replace}
	Let $Q_i$ consist of edges in $E_\alpha$ ($\alpha$ = 1 or 2).
	$Q_i$ can be replaced by an edge from $r_{i-1}$ to $r_i$ where $r_{i-1}r_i \in E_\alpha$.
\end{lemma}
\begin{proof}
	A closed subset separating $r_{i-1}r_i$ must separate an edge in $Q_i$.
	Moreover, any closed subset must separate exactly one of $r_{0}r_{1}, \ldots, r_{k-2}r_{k-1} , r_{k-1}r_k$.
	Therefore, the set of closed subsets separating an edge in $E_1$ (or $E_2$) remains unchanged.
\end{proof}

\begin{lemma}
	\label{lem:remove}
	Edges in $E_1\union E_2$ but not in $Q$ can be removed.
\end{lemma}
\begin{proof}
	Let $e$ be an edge in $E_1\union E_2$ but not in $Q$. 
	Suppose that $e \in E_1$.
	Let $I$ be a closed subset separating $e$.
	By Lemma~\ref{lem:separating}, the matching generated by $I$ belongs to $\Lc_2$.
	Since $e$ is not in $Q$ and $Q$ is a path from $t$ to $s$, $I$ must separate another edge $e'$ in $Q$. 
	By Lemma~\ref{lem:separating}, $I$ can not separate edges in both $E_1$ and $E_2$. 
	Therefore, $e'$ must also be in $E_1$. 
	Hence, the matching generated by $I$ will still be in $\Lc_2$ after removing $e$ from $E_1$.
	The argument applies to all closed subsets separating $e$.
\end{proof}

By Lemma~\ref{lem:replace} and Lemma~\ref{lem:remove}, $r_{0}r_{1}, \ldots, r_{k-2}r_{k-1} , r_{k-1}r_k$ are all edges in $E_1$ and $E_2$ and they alternate between $E_1$ and $E_2$. Therefore, we have Theorem~\ref{thm:alternating}. An illustration of such a path is given in Figure~\ref{ex:canonicalPathAndBouquet}(a).


\begin{proposition}
\label{prop.sub}
	There exists a sequence of rotations $r_0, r_1, \ldots , r_{2k}, r_{2k+1}$ such that 
	a closed subset generates a matching in  $\Lc_1$ iff it
	contains $r_{2i}$ but not $r_{2i+1}$ for some $0 \leq i \leq k$.
\end{proposition}

\begin{figure}
	\begin{wbox}
		\begin{minipage}[c]{0.5\textwidth}
			\centering
			\def\svgscale{0.4}
			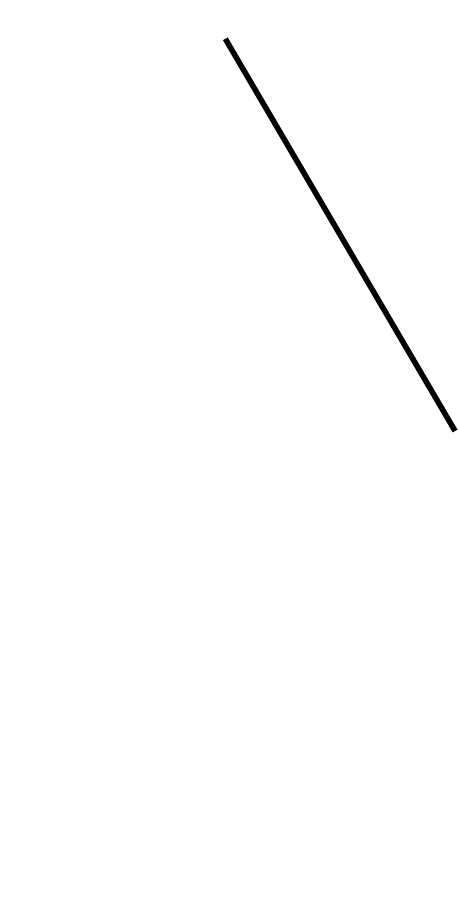
			
			(a)
		\end{minipage}
		\begin{minipage}[c]{0.5\textwidth}
			\centering
			\def\svgscale{0.4}
			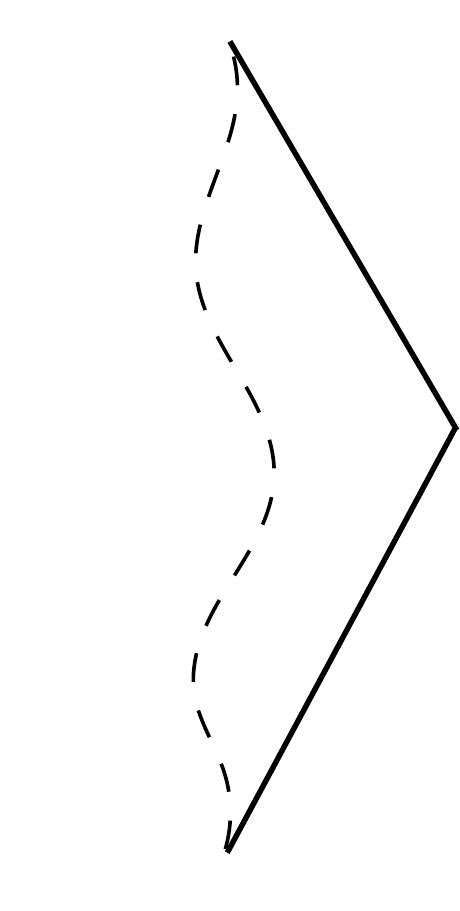
			
			(b)	
		\end{minipage}
	\end{wbox}
	\caption{Examples of: (a) canonical path, and (b) bouquet. 
	}
	\label{ex:canonicalPathAndBouquet} 
\end{figure}

\section{Setting II}
\label{sec.semi}

Under Setting II, the given lattice $\Lc$ can be partitioned into a sublattice $\Lc_1$ and a 
semi-sublattice $\Lc_2$. 
We assume that $\Lc_2$ is a join semi-sublattice. 
Clearly by reversing the order of $\Lc$, the case of meet semi-sublattice is also covered.
The next theorem, which generalizes Theorem \ref{thm:alternating},
gives a sufficient characterization of a set of edges $E$ defining $\Lc_1$.

\begin{theorem}
\label{thm:semi}
There exists a set of edges $E$ defining sublattice $\Lc_1$ such that:
\begin{enumerate}
	\item The set of tails $T_E$ of edges in $E$ forms a chain in $\Pi$.
	\item There is no path of length two consisting of edges in $E$.
	\item For each $r \in T_E$, let
	\[F_r = \{ v \in \Pi: rv \in E \}.\]
	Then any two rotations in $F_r$ are incomparable.
	\item For any $r_i,r_j \in T_E$ where $r_i \prec r_j$, there exists a splitting set containing all rotations in $F_{r_i} \union \{r_i\}$ and no rotations in $F_{r_j} \union \{r_j\}$.  
\end{enumerate}
\end{theorem}

A set $E$ satisfying Theorem~\ref{thm:semi} will be called a \emph{bouquet}. 
For each $r \in T_E$, let $L_r = \{rv \ | \ v \in F_r \}$. Then $L_r$ will be called a 
\emph{flower}. Observe that the bouquet $E$ is partitioned into flowers. 
These notions are illustrated in Figure~\ref{ex:canonicalPathAndBouquet}(b). The black path,
directed from $s$ to $t$, is the chain mentioned in Theorem \ref{thm:semi} and
the red edges constitute $E$. Observe that the tails of edges $E$ lie on the chain. For each such
tail, the edges of $E$ outgoing from it constitute a flower.

Let $E$ be an arbitrary set of edges defining $\Lc_1$. We will show that $E$ can be modified so 
that the conditions in Theorem~\ref{thm:semi} are satisfied.
Let $S$ be a splitting set of $\Pi$. In other words, $S$ is a closed subset such that for all $uv \in E$, either $u,v$ are both in $S$ or $u,v$ are both in $\Pi \setminus S$.

\begin{lemma}
	\label{lem:uniqueMaximal}
	There is a unique maximal rotation in $T_E \intersect S$.
\end{lemma}
\begin{proof}
	Suppose there are at least two maximal rotations $u_1,u_2, \ldots u_k$ ($k \geq 2$) in $T_E \intersect S$.
	Let $v_1, \ldots v_k$ be the heads of edges containing  $u_1,u_2, \ldots u_k$.
	For each $1 \leq i \leq k$, let $S_i = J_{u_i} \union J_{v_j}$ where $j$ is any index such that $j \not = i$.
	Since $u_i$ and $u_j$ are incomparable, $u_j \not \in J_{u_i}$. 
	Moreover, $u_j \not \in J_{v_j}$ by Lemma~\ref{lem:prime}.  
	Therefore, $u_j \not \in S_i$. 
	It follows that $S_i$ contains $u_i$ and separates $u_j v_j$.
	Since $S_i$ separates $u_jv_j \in E$, the matching generated by $S_i$ is in $\Lc_2$ according to Lemma~\ref{lem:separating}.
	
	Since $\Union_{i=1}^k S_i$ contains all maximal rotations in $T_E \intersect S$ and $S$ does not separate any edge in $E$,
	$\Union_{i=1}^k S_i$ does not separate any edge in $E$ either. Therefore, the matching generated by $\Union_{i=1}^k S_i$ 
	is in $\Lc_1$, and hence not in $\Lc_2$. This contradicts the fact that $\Lc_2$ is a join semi-sublattice. 
\end{proof}

Denote by $r$ the unique maximal rotation in $T_E \intersect S$. Let
\begin{align*}
R_r &= \{ v \in \Pi: \text{ there is a path from $r$ to $v$ using edges in $E$}\}, \\
E_r &= \{ uv \in E: u,v \in R_r \}, \\
G_r &= \{R_r, E_r \}.
\end{align*}
Note that $r \in R_r$.
For each $v \in R_r$ there exists a path from $r$ to $v$ and $r \in S$. 
Since $S$ does not cross any edge in the path,  $v$ must also be in $S$.
Therefore, $R_r \subseteq S$. 

\begin{lemma}
	\label{lem:semi-replace}
	Let $u \in (T_E \intersect S) \setminus R_{r}$ such that $u \succ x$ for $x \in R_{r}$. Then we can replace each $uv \in E$ with $rv$.
\end{lemma}
\begin{proof}
	We will show that the set of closed subsets separating an edge in $E$ remains unchanged. 
	
	Let $I$ be a closed subset separating $uv$. Then $I$ must also separate $rv$ since $r \succ v$. 
	
	Now suppose $I$ is a closed subset separating $rv$. We consider two cases:
	\begin{itemize}
		\item If $u \in I$, $I$ must contain $x$ since $u \succ x$. Hence, $I$ separates an edge in the path from $r$ to $x$.
		\item If $u \not \in I$, $I$ separates $uv$. 
	\end{itemize}
\end{proof}

Keep replacing edges according to Lemma~\ref{lem:semi-replace} until there is no $u \in (T_E \intersect S) \setminus R_{r}$ such that $u \succ x$ for some $x \in R_{r}$.

\begin{lemma}
	\label{lem:flower-separating}
	Let
	\[ X = \{v \in S: v \succeq x \text{ for some } x \in R_{r}\}. \]
	\begin{enumerate}
		\item $S \setminus X$ is a closed subset. 
		\item $S \setminus X$ contains $u$ for each $u \in (T_E \intersect S) \setminus R_{r}$.
		\item $(S \setminus X) \intersect R_{r} = \emptyset$.
		\item $S \setminus X$ is a splitting set.
	\end{enumerate}
\end{lemma}

\begin{proof}
The lemma follows from the claims given below: 

	\paragraph{Claim 1.}  $S \setminus X$ is a closed subset.
	\begin{proof}
		Let $v$ be a rotation in $S \setminus X$ and $u$ be a predecessor of $v$.
		Since $S$ is a closed subset, $u \in S$.
		Notice that if a rotation is in $X$, all of its successor must be included.
		Hence, since $v \notin X$, $u \notin X$.
		Therefore, $u \in S \setminus X$.
	\end{proof}
	
	\paragraph{Claim 2.} $S \setminus X$ contains $u$ for each $u \in (T_E \intersect S) \setminus R_{r}$.
	\begin{proof}
		After replacing edges according to Lemma~\ref{lem:semi-replace}, for each $u \in (T_E \intersect S) \setminus R_{r}$ we must have that $u$ does not succeed any $x \in R_r$. Therefore, $u \notin X$ by the definition of $X$.		
	\end{proof}

	\paragraph{Claim 3.} $(S \setminus X) \intersect R_{r} = \emptyset$.
	\begin{proof}
		Since $R_{r} \subseteq X$, $(S \setminus X) \intersect R_{r} = \emptyset$.
	\end{proof}
	
	\paragraph{Claim 4.} $S \setminus X$ does not separate any edge in $E$.
	\begin{proof}
		Suppose $S \setminus X$ separates $uv \in E$. Then $u \in X$ and $v \in  S \setminus X$.
		By Claim 2, $u$ can not be a tail vertex, which is a contradiction.  
	\end{proof}
	
	\paragraph{Claim 5.} $S \setminus X$ does not cross any edge in $E$.
	\begin{proof}
		Suppose $S \setminus X$ crosses $uv \in E$. Then $u \in S \setminus X$ and $v \in X$.
		Let $J$ be a closed subset separating $uv$. Then $v \in J$ and $u \notin J$.
		
		Since $uv \in E$ and $u \in S$, $u \in T_E \intersect S$. Therefore, $r \succ u$ by Lemma~\ref{lem:uniqueMaximal}.
		Since $J$ is a closed subset, $r \notin J$.
		
		Since $v \in X$, $v \succeq x$ for $x \in R_r$. Again, as $J$ is a closed subset, $x \in J$.  
		
		Therefore, $J$ separates an edge in the path from $r$ to $x$ in $G_r$. 
		Hence, all closed subsets separating $uv$ must also separate another edge in $E_r$.
		This contradicts the assumption made in Remark~\ref{remark:addedEdge}.
	\end{proof}
\end{proof}

\begin{lemma}
	\label{lem:semi-replace2}
	$E_r$ can be replaced by the following set of edges:
	\[E'_r = \{rv: v \in R_r \}.\]
\end{lemma}
\begin{proof}

	We will show that the set of closed subsets separating an edge in $E_r$ and 
	the set of closed subset separating an edge in $E'_r$ are identical. 
	
	Consider a closed subset $I$ separating an edge in $rv \in E'_r$. 
	Since $v \in R_r$, $I$ must separate an edge in $E$ in a path from $r$ to $v$.
	By definition, that edge is in $E_r$.
	
	Now let $I$ be a closed subset separating an edge in $uv \in E_r$.
	Since $uv \in E$, $u \in T_E \intersect S$. By Lemma~\ref{lem:uniqueMaximal}, $r \succ u$. 
	Thus, $I$ must also separate $rv \in E'_r$. 
\end{proof}

\begin{proof} [Proof of Theorem~\ref{thm:semi}]
	To begin, let $S_1 = \Pi$ and let $r_1$ be the unique maximal rotation according to Lemma~\ref{lem:uniqueMaximal}.
	Then we can replace edges according to Lemma~\ref{lem:semi-replace} and Lemma~\ref{lem:semi-replace2}.
	After replacing, $r_1$ is the only tail vertex in $G_{r_1}$. 
	By Lemma~\ref{lem:flower-separating}, there exists a set $X$ such that $S_1 \setminus X$ does not contain any vertex in $R_{r_1}$ and contains all other tail vertices in $T_E$ except $r_1$. 
	Moreover, $S_1 \setminus X$ is a splitting set.
	Hence, we can set $S_2 = S_1 \setminus X$ and repeat. 
	
	Let $r_1, \ldots , r_k$ be the rotations found in the above process. Since $r_i$ is the unique maximal rotation in $T_E \intersect S_i$ for all $1 \leq i \leq k$ and $S_1 \supset S_2 \supset \ldots \supset S_k$, we have $r_1 \succ r_2 \succ \ldots \succ r_k$. 
	By Lemma~\ref{lem:semi-replace2}, for each $1 \leq i \leq k$, $E_{r_i}$ consists of edges $r_iv$ for $v \in R_{r_i}$.
	Therefore, there is no path of length two composed of edges in $E$ and condition 2 is satisfied. Moreover, $r_1, \ldots , r_k$ are exactly the tail vertices in $T_E$, which gives condition 1.
	
	Let $r$ be a rotation in $T_E$ and consider $u,v \in F_r$. Moreover, assume that $u \prec v$. A closed subset $I$ separating $rv$ contains $v$ but not $r$. Since $I$ is a closed subset and $u \prec v$, $I$ contains $u$. Therefore, $I$ also separates $ru$, contradicting the assumption in Remark~\ref{remark:addedEdge}. The same argument applies when $v \prec u$. Therefore, $u$ and $v$ are incomparable as stated in condition 3. 
	
	Finally, let $r_i, r_j \in T_E$ where $r_i \prec r_j$. By the construction given above, 
	$S_j \supset S_{j-1} \supset \ldots \supset S_i$, $R_{r_j} \subseteq S_j \setminus S_{j-1}$ and $R_{r_i} \subseteq S_i$.
	Therefore, $S_i$ contains all rotations in $R_{r_i}$ but none of the rotations in $R_{r_j}$, giving condition 4. 
\end{proof}

\begin{proposition}
\label{prop.gen}
	There exists a sequence of rotations $r_1 \prec \ldots \prec r_{k}$ 
	and a set $F_{r_i}$ for each $1 \leq i \leq k$ such that 
	a closed subset generates a matching in $\Lc_1$ if and only if
	whenever it contains a rotation in $F_{r_i}$, it must also contain $r_i$.
\end{proposition}

	\section{Algorithm for Finding a Bouquet}
\label{sec.alg}
In this section, we give an algorithm for finding a bouquet. Let $\Lc$
be a distributive lattice that can be partitioned into a sublattice $\Lc_1$ and a 
semi-sublattice $\Lc_2$.
Then given a poset $\Pi$ of $\Lc$ and a membership oracle, which determines if a matching
of $\Lc$ is in $\Lc_1$ or not, the algorithm returns a bouquet defining $\Lc_1$.

By Theorem~\ref{thm:semi}, the set of tails $T_E$ forms a chain $C$ in $\Pi$.
The idea of our algorihm, given in Figure~\ref{alg:flowerSet}, is to find 
the flowers according to their order in $C$.
Specifically, a splitting set $S$ is maintained such that at any point, all flowers outside of $S$ are found.
At the beginning, $S$ is set to $\Pi$ and becomes smaller as the algorithm proceeds. 
Step 2 checks if $M_z$ is a matching in $\Lc_1$ or not. 
If $M_z \not\in \Lc_1$, the closed subset $\Pi \setminus \{t\}$ separates an edge in 
$E$ according to Lemma~\ref{lem:separating}. Hence, the first tail on $C$ must be $t$.
Otherwise, the algorithm jumps to Step 3 to find the first tail. 
Each time a tail $r$ is found, Step 5 immediately finds the flower $L_r$ corresponding to $r$. 
The splitting set $S$ is then updated so that $S$ no longer contains $L_r$ 
but still contains the flowers that have not been found yet. 
Next, our algorithm continues to look for the next tail inside the updated $S$.
If no tail is found, it terminates. 

\begin{figure}
	\begin{wbox}
		\textsc{FindBouquet}$(\Pi)$: \\
		\textbf{Input:} A poset $\Pi$. \\
		\textbf{Output:} A set $E$ of edges defining $\Lc_1$. 
		\begin{enumerate}
			\item Initialize: Let $S = \Pi, E = \emptyset$.
			\item If $M_z$ is in $\Lc_1$: go to Step 3. Else: $r = t$, go to Step 5.
			\item $r$ = \textsc{FindNextTail}$(\Pi,S)$. 
			\item If $r$ is not \textsc{Null}: Go to Step 5. Else: Go to Step 7. 
			\item $F_r$ = \textsc{FindFlower}$(\Pi,S,r)$.
			\item Update:
			\begin{enumerate}
				\item For each $u \in F_r$: $E \leftarrow E \union \{ru\}$.
				\item $ S \leftarrow S \setminus  \Union_{u \in F_r \union \{ r \}} J'_u $.
				\item Go to Step 3.
			\end{enumerate}
			\item Return $E$. 
		\end{enumerate}
	\end{wbox}
	\caption{Algorithm for finding a bouquet.}
	\label{alg:flowerSet} 
\end{figure}

First we prove a simple observation.

\begin{lemma}
	\label{lem:headTailCondition}
	Let $v$ be a rotation in $\Pi$. Let $S \subseteq \Pi$ such that both $S$ and $S \union \{v\}$ are closed subsets.
	If $S$ generates a matching in $\Lc_1$ and $S \union \{v\} $ generates a matching in $\Lc_2$, 
		$v$ is the head of an edge in $E$.
	If $S$ generates a matching in $\Lc_2$ and $S \union \{v\} $ generates a matching in $\Lc_1$, 
		$v$ is the tail of an edge in $E$.

\end{lemma}
\begin{proof}
	Suppose that $S$ generates a matching in $\Lc_1$ and $S \union \{v\} $ generates a matching in $\Lc_2$.
	By Lemma~\ref{lem:separating}, $S$ does not separate any edge in $E$, and 
	$S \union \{v\}$ separates an edge $e \in E$. 
	This can only happen if $u$ is the head of $e$.
	
	A similar argument can be given for the second case.
\end{proof}

\begin{figure}
	\begin{wbox}
		\textsc{FindNextTail}$(\Pi,S)$: \\
		\textbf{Input:} A poset $\Pi$, a splitting set $S$. \\
		\textbf{Output:} The maximal tail vertex in $S$, or \textsc{Null} if there is no tail vertex in $S$.
		\begin{enumerate}
			\item Compute the set $V$ of rotations $v$ in $S$ such that:  
			\begin{itemize}
				\item $\Pi \setminus I'_v$ generates a matching in $\Lc_1$.
				\item $\Pi \setminus J'_v$ generates a matching in $\Lc_2$.
			\end{itemize}
			\item If $V \not = \emptyset$ and  there is a unique maximal element $v$ in $V$: Return $v$. \\
			Else: Return \textsc{Null}.
		\end{enumerate}
	\end{wbox}
	\caption{Subroutine for finding the next tail.}
	\label{alg:findTail} 
\end{figure} 

\begin{lemma}
	\label{lem:correctnessFindNextTail}
	Given a splitting set $S$,
	\textsc{FindNextTail}$(\Pi,S)$ (Figure~\ref{alg:findTail}) returns the maximal tail vertex in $S$, 
	or \textsc{Null} if there is no tail vertex in $S$.
\end{lemma}
\begin{proof}
	Let $r$ be the maximal tail vertex in $S$.

	First we show that $r \in V$.
	By Theorem~\ref{thm:semi}, the set of tails of edges in $E$ forms a chain in $\Pi$.
	Therefore $\Pi \setminus I'_r$ contains all tails in $S$. 
	Hence, $\Pi \setminus I'_r$ does not separate any edge whose tails are in $S$. 
	Since $S$ is a splitting set, $\Pi \setminus I'_r$ does not separate any edge whose tails are in $\Pi \setminus S$.   
	Therefore, by Lemma~\ref{lem:separating}, $\Pi \setminus I'_r$ generates a matching in $\Lc_1$.
	By Lemma~\ref{lem:prime}, $\Pi \setminus J'_r$ must separate an edge in $E$, 
	and hence generates a matching in $\Lc_2$ according to Lemma~\ref{lem:separating}. 
	
	By Lemma~\ref{lem:headTailCondition}, any rotation in $V$ must be the tail of an edge in $E$.
	Hence, they are all predecessors of $r$ according to Theorem~\ref{thm:semi}.
\end{proof}

\begin{figure}
	\begin{wbox}
		\textsc{FindFlower}$(\Pi,S,r)$: \\
		\textbf{Input:} A poset $\Pi$, a tail vertex $r$ and a splitting set $S$ containing $r$. \\
		\textbf{Output:} The set $F_r = \{ v \in \Pi: rv \in E \}$.
		\begin{enumerate}
			\item Compute $X = \{ v \in I_r: J_v \text{ generates a matching in } \Lc_1 \}$.
			\item Let $Y = \Union_{v \in X} J_v$.
			\item If $Y = \emptyset$ and $M_0 \in \Lc_2$: Return $\{s\}$. 
			\item Compute the set $V$ of rotations $v$ in $S$ such that:
			\begin{itemize}
				\item $Y \union I_v$ generates a matching in $\Lc_1$.
				\item $Y \union J_v$ generates a matching in $\Lc_2$.
			\end{itemize} 
			\item Return $V$.
		\end{enumerate}
	\end{wbox}
	\caption{Subroutine for finding a flower.}
	\label{alg:findflower} 
\end{figure}

\begin{lemma}
	\label{lem:correctnessFindFlower}
	Given a tail vertex $r$ and a splitting set $S$ containing $r$, \textsc{FindFlower}$(\Pi,S,r)$ (Figure~\ref{alg:findflower})
	correctly returns $F_r$.
\end{lemma}

\begin{proof}
	First we give two crucial properties of the set $Y$.
	By Theorem~\ref{thm:semi}, the set of tails of edges in $E$ forms a chain $C$ in $\Pi$.
	
	\paragraph{Claim 1.} $Y$ contains all predecessors of $r$ in $C$.
	\begin{proof}
		Assume that there is at least one predecessor of $r$ in $C$, and denote by $r'$ the direct predecessor. 
		It suffices to show that $r' \in Y$. 
		By Theorem~\ref{thm:semi}, there exists a splitting set $I$ such that $R_{r'} \subseteq I$ and $R_r \intersect I = \emptyset$.
		Let $v$ be the maximal element in $C \intersect I$.
		Then $v$ is a successor of all tail vertices in $I$.
		It follows that $J_v$ does not separate any edges in $E$ inside $I$.
		Therefore, $v \in X$.
		Since $J_v \subseteq Y$, $Y$ contains all predecessors of $r$ in $C$.	
	\end{proof} 
	
	\paragraph{Claim 2.}  $Y$ does not contain any rotation in $F_r$.
	\begin{proof}
		Since $Y$ is the union of closed subset generating matching in $\Lc_1$, $Y$ also generates a matching in $\Lc_1$.
		By Lemma~\ref{lem:separating}, $Y$ does not separate any edge in $E$. 
		Since $r \not \in Y$, $Y$ must not contain any rotation in $F_r$.
	\end{proof}
	
	By Claim 1, if $Y = \emptyset$, $r$ is the last tail found in $C$. 
	Hence, if $M_0 \in \Lc_2$, $s$ must be in $F_r$. 
	By Theorem~\ref{thm:semi}, the heads in $F_r$ are incomparable. 
	Therefore, $s$ is the only rotation in $C$.
	\textsc{FindFlower} correctly returns $\{s\}$ in Step 3.
	Suppose such a situation does not happen, we will show that the returned set is $F_r$.

	\paragraph{Claim 3.} $V = F_r$.
	\begin{proof}
		Let $v$ be a rotation in $V$. 
		By Lemma~\ref{lem:headTailCondition}, $v$ is a head of some edge $e$ in $E$.
		Since $Y$ contains all predecessors of $r$ in $C$, the tail of $e$ must be $r$. 
		Hence, $v \in  F_r$.
		
		Let $v$ be a rotation in $F_r$. 
		Since $Y$ contains all predecessors of $r$ in $C$, $Y \union I_v$ can not separate any edge 
		whose tails are predecessors of $r$. 
		Moreover, by Theorem~\ref{thm:semi}, the heads in $F_r$ are incomparable. 
		Therefore, $I_v$ does not contain any rotation in $F_r$. 
		Since $Y$ does not contain any rotation in $F_r$ by the above claim, 
		$Y \union I_v$ does not separate any edge in $E$.
		It follows that $Y \union I_v$ generates a matching in $\Lc_1$.
		Finally, $Y \union J_v$ separates $rv$ clearly, and hence generates a matching in $\Lc_2$.
		Therefore, $v \in V$ as desired. 
	\end{proof}
\end{proof}

\begin{theorem}
	\label{thm:algFindFlower}
	\textsc{FindBouquet}$(\Pi)$, given in Figure~\ref{alg:flowerSet}, 
	returns a set of edges defining $\Lc_1$.
\end{theorem}
\begin{proof}
	From Lemmas~\ref{lem:correctnessFindNextTail} and \ref{lem:correctnessFindFlower}, it suffices to show that
	$S$ is udpated correctly in Step 6(b). To be precised, we need that
	\[ S \setminus  \Union_{u \in F_r \union \{ r \}} J'_u \]
	must still be a splitting set, and contains all flowers that have not been found. 
	This follows from Lemma~\ref{lem:flower-separating} by noticing that 
	\[ \Union_{u \in F_r \union \{ r \}} J'_u  = \{v \in \Pi: v \succeq u \text{ for some } u \in R_r\}.\]
\end{proof}

Clearly, a sublattice of $\Lc$ must also be a semi-sublattice. Therefore, \textsc{FindBouquet} can be used to find a canonical path described in Section~\ref{sec.sublattice}. The same algorithm can be used to check if $M_A \cap M_B = \emptyset$. Let $E$ be the edge set given by the \textsc{FindBouquet} algorithm and $H_E$ be the corresponding graph obtained by adding $E$ to the Hasse diagram of the original rotation poset $\Pi$ of $\Lc_A$. If $H_E$ has a single strongly connected component, the compression $\Pi'$ has a single meta-element and represents the empty lattice.

	\section{Finding an Optimal Fully Robust Stable Matching}
\label{sec.robust}
Consider the setting given in the Introduction, with $S$ being the domain of errors, one of
which is introduced in instance $A$.
We show how to use the algorithm in Section~\ref{sec.alg} to 
find the poset generating all fully robust matchings w.r.t. $S$. We 
then show how this poset can yield a fully robust matching that maximizes, or minimizes, 
a given weight function. 

\subsection{Studying semi-sublattices is necessary and sufficient}
\label{subsection.semisublatticeNecessarySufficient}

Let $A$ be a stable matching instance, and $B$ be an instance obtained by permuting
the preference list of one worker or one firm. 
Lemma \ref{lem.eg} gives an example
of a permutation so that $\Mc_{A \setminus B}$ is not a sublattice of $\Lc_A$, hence showing that
the case studied in Section \ref{sec.sublattice} does not suffice to solve the problem at hand.
On the other hand, for all such instances $B$, Lemma \ref{lem:MABsemi} shows that
$\Mc_{A \setminus B}$ forms a semi-sublattice of $\Lc_A$ and hence the case studied in Section
\ref{sec.semi} does suffice.

The next lemma pertains to the example given in Figure~\ref{ex:notSublattice},
in which the set of workers is  $\mathcal{B} = \{a,b,c,d\}$ and the set of firms is 
$\mathcal{G} = \{1,2,3,4\}$. Instance $B$ is obtained from instance $A$ by permuting firm 1's 
list.

\begin{lemma} 
\label{lem.eg}
	$\Mc_{A \setminus B}$ is not a sublattice of $\Lc_A$.
\end{lemma}
\begin{proof}
	$M_1 = \{1a,2b,3d,4c\}$ and $M_2 = \{1b,2a,3c,4d\}$ are stable matching with respect to instance $A$.
	Clearly, $M_1 \wedge_A M_2 = \{1a,2b,3c,4d\}$ is also a stable matching under $A$.
	
	In going from $A$ to $B$, the positions of workers $b$ and $c$ are swapped in firm 1's list. 
	Under $B$, $1c$ is a blocking pair for $M_1$ and $1a$ is a blocking pair for $M_2$.
	Hence, $M_1$ and $M_2$ are both in $\Mc_{A \setminus B}$.
However, $M_1 \wedge_A M_2$ is a stable matching under $B$, and therefore is it not in $\Mc_{A \setminus B}$.
Hence, $\Mc_{A \setminus B}$ is not closed under the $\wedge_A$ operation.
\end{proof}

\begin{figure}
	\begin{wbox}
	\begin{minipage}{.33\linewidth}
		\centering
		\begin{tabular}{l|llll}
			1  & b & a & c & d \\
			2  & a & b & c & d \\
			3  & d & c & a & b \\
			4  & c & d & a & b
		\end{tabular}
		
		\hspace{1cm}
		
		firms' preferences in $A$ 
		
		\hspace{1cm}
		
	\end{minipage}%
	\begin{minipage}{.34\linewidth}
	\centering
	\begin{tabular}{l|llll}
		\color{red}{1}  & \color{red}{c} & \color{red}{a} & \color{red}{b} & \color{red}{d} \\
		2  & a & b & c & d \\
		3  & d & c & a & b \\
		4  & c & d & a & b
	\end{tabular}
	
	\hspace{1cm}
	
	firms' preferences in $B$
	 
	\hspace{1cm}
\end{minipage}%
	\begin{minipage}{.33\linewidth}
		\centering
		\begin{tabular}{l|llll}
		a  & 1 & 2 & 3 & 4 \\
		b  & 2 & 1 & 3 & 4 \\
		c  & 3 & 1 & 4 & 2 \\
		d  & 4 & 3 & 1 & 2 
		\end{tabular}

		\hspace{1cm}

		workers' preferences in both instances
	\end{minipage} 
	\end{wbox}
	\caption{An example in which $\Mc_{A \setminus B}$ is not a sublattice of $\Lc_A$.}
	\label{ex:notSublattice} 
\end{figure}

\begin{lemma}
	\label{lem:MABsemi}
	For any instance $B$ obtained by permuting the preference list of one worker or one firm,
	$\Mc_{A \setminus B}$ forms a semi-sublattice of $\Lc_A$. 
\end{lemma}
\begin{proof}
	Assume that the preference list of a firm $f$ is permuted. 
	We will show that $\Mc_{A \setminus B}$ is a join semi-sublattice of $\Lc_A$.
	By switching the role of workers and firms, 
	permuting the list of a worker will result in $\Mc_{A \setminus B}$ 
	being a meet semi-sublattice of $\Lc_A$.
	
	Let $M_1$ and $M_2$ be two matchings in $\Mc_{A \setminus B}$. 
	Hence, neither of them are in $\Mc_{B}$. 
	In other words, each has a blocking pair under instance $B$.
	
	Let $w$ be the partner of $f$ in $M_1 \vee_A M_2$.
	Then $w$ must also be matched to $f$ in either $M_1$ or $M_2$ (or both).
	We may assume that $w$ is matched to $f$ in $M_1$. 
	
	Let $xy$ be a blocking pair of $M_1$ under $B$. 
	We will show that $xy$ must also be a blocking pair of $M_1 \vee_A M_2$ under $B$.
	To begin, the firm $y$ must be $f$ since other preference lists remain unchanged. 
	Since $xf$ is a blocking pair of $M_1$ under $B$, $x >_f^B w$.
	Similarly, $f >_x f'$ where $f'$ is the $M_1$-partner of $x$. 
	Let $f''$ be the partner of $x$ in $M_1 \vee_A M_2$.
	Then $f' \geq_x f''$. It follows that $f >_x f''$.
	Since $x >_f^B w$ and $f >_x f''$,
	$xf$ must be a blocking pair of $M_1 \vee_A M_2$ under $B$.
\end{proof}


\begin{proposition}
	\label{prop:MABcompute}
A set of edges defining the sublattice $\Lc'$, consisting of matchings in $\Mc_A \intersect \Mc_B$,
	can be computed efficiently.
\end{proposition}
\begin{proof}
	 We have that $\Lc'$ and $\Mc_{A \setminus B}$ partition $\Lc_A$,
	 with $\Mc_{A \setminus B}$ being a semi-sublattice of $\Lc_A$, by Lemma~\ref{lem:MABsemi}.
	 Therefore,	\textsc{FindBouquet}$(\Pi)$
	 finds a set of edges defining $\Lc'$ by Theorem~\ref{thm:algFindFlower}.
	
	By Lemma~\ref{lem:computePoset}, the input $\Pi$ to \textsc{FindBouquet} can be computed in polynomial time. 
	Clearly, a membership oracle checking if a matching is in $\Lc'$ or 
	not can also be implemented efficiently. 
	Since $\Pi$ has $O(n^2)$ vertices (Lemma~\ref{lem:computePoset}), any step of \textsc{FindBouquet} takes polynomial time.	
\end{proof}

\subsection{Proof of Theorem \ref{thm:main}}
\label{sec.fullyRobust} 

In this section, we will prove Theorem \ref{thm:main} as well as a slight extension; the latter uses ideas from \cite{MV.robust}. Let $B_1, \ldots, B_k$ be polynomially many instances in the domain $D \subset T$,
as defined in the Introduction.
Let $E_{i}$ be the set of edges defining $\Mc_A \intersect \Mc_{B_i}$ for all $1 \leq i \leq k$.
By Corollary~\ref{cor.sublatticeIntersection}, $\Lc' = \Mc_A \intersect \Mc_{B_1} \intersect \ldots \intersect \Mc_{B_k}$ is a sublattice of $\Lc_A$.

\begin{lemma}
	\label{lem:poset}
	$E = \Union_i E_{i}$ defines $\Lc'$.
\end{lemma}
\begin{proof}
	By Lemma~\ref{lem:separating}, it suffices to show that for any closed subset $I$,
	$I$ does not separate an edge in $E$ iff $I$ generates a matching in $\Lc'$.
	
	$I$ does not separate an edge in $E$ iff
	$I$ does not separate any edge in $E_i$  for all $1 \leq i\leq k$ iff
	the matching generated by $I$ is in 
	$\Mc_A \intersect \Mc_{B_i}$ for all $1 \leq i \leq k$ by Lemma~\ref{lem:separating}. 
\end{proof}

By Lemma~\ref{lem:poset}, a compression $\Pi'$ generating $\Lc'$ can be constructed from $E$ as described in Section~\ref{sec.alternative}.
By Proposition~\ref{prop:MABcompute}, we can compute each $E_{i}$, and hence, $\Pi'$ efficiently. 
Clearly, $\Pi'$ can be used to check if a fully robust stable matching exists. To be precise, 
a fully robust stable matching exists iff there exists a proper closed subset of $\Pi'$.
This happens iff $s$ and $t$ belong to different meta-rotations in $\Pi'$, an easy to check
condition. Hence, we have Theorem~\ref{thm:main}. 

We can use $\Pi'$ to obtain a fully robust stable matching $M$ maximizing
$ \sum_{wf \in M} W_{wf}$
by applying the algorithm of \cite{MV.weight}.
Specifically, let $H(\Pi')$ be the Hasse diagram of $\Pi'$. 
Then each pair $wf$ for $w \in \mathcal{W}$ and $f \in \mathcal{F}$ 
can be associated with two vertices $u_{wf}$ and $v_{wf}$ in $H(\Pi')$ as follows:
\begin{itemize}
	\item If there is a rotation $r$ moving $w$ to $f$, $u_{wf}$ is the meta-rotation containing $r$. Otherwise, $u_{wf}$ is the meta-rotation containing $s$.
	\item If there is a rotation $r$ moving $w$ from $f$, $v_{wf}$ is the meta-rotation containing $r$. Otherwise, $v_{wf}$ is the meta-rotation containing $t$.
\end{itemize}

By Lemma~\ref{lem:pre2} and the definition of compression, $u_{wf} \prec v_{wf}$.
Hence, there is a path from $u_{wf}$ to $v_{wf}$ in $H(\Pi')$.
We can then add weights to edges in $H(\Pi')$, as stated in \cite{MV.weight}.
Specifically, we start with weight 0 on all edges and increase weights of edges in a path from 
$u_{wf}$ to $v_{wf}$ by $w_{wf}$ for all pairs $wf$.
A fully robust stable matching maximizing
$ \sum_{wf \in M} W_{bwf}$ can be obtained by 
finding a maximum weight ideal cut in the constructed graph. An efficient algorithm for the 
latter problem is given in \cite{MV.weight}.

    \section{Discussion}
\label{sec.discussion}

A number of new questions arise: give a polynomial time algorithm for the problem mentioned in the Introduction, of finding a robust stable matching as defined in \cite{MV.robust}  --- given a probability distribution on the domain of errors --- even when the error is an arbitrary permutation; extend the results to multiple agents simultaneously changing their preference lists; and extend to the stable matching problem with incomplete preference lists, the stable roommate problem \cite{GusfieldI, Manlove-book}, and popular matchings \cite{2popMatching, kavitha2021matchings}.

Next, we give a hypothetical setting to show potential application of our work to the issue of incentive compatibility. Let $A$ be an instance of stable matching over $n$ workers and $n$ firms. Assume that all $2n$ agents have a means of making their preference lists public simultaneously and a dominant firm, say $f$, is given the task of computing and announcing a stable matching. Once the matching is announced, all agents can verify that it is indeed stable. It turns out that firm $f$ can cheat and improve its match as follows: $f$ changes its preference list to obtain instance $B$ which is identical to $A$ for all other agents, and computes a matching that is stable for $A$ as well as $B$ using Theorem \ref{thm:main}. The other agents will be satisfied that this matching is indeed stable for instance $A$ and $f$'s cheating may go undetected.  

The problem solved in this paper appears to be a basic one and therefore ought to have an impactful application. One avenue that may lead to it is the number of new and interesting matching markets being defined on the Internet, e.g., see \cite{Simons}.

	\bibliographystyle{alpha}
	\bibliography{refs}
	
	\appendix
	
	\section{Modified Deferred Acceptance Algorithms}
\label{app:algorithms}

\subsection{Errors on only the Firms Side}

In the case that only one side changes preferences, (Corollary~\ref{cor.sublatticeIntersection}) shows that the set of fully robust stable matchings is a sublattice of the original stable matching lattice ($L_A$). 

Firstly, assume only firms are allowed to change preferences. Then, the modification of the Deferred Acceptance Algorithm (Algorithm~\ref{alg:daalgorithm}) can be used to find a fully robust stable matching. This algorithm is equivalent to the Algorithm 1 in \cite{genPreferences} which finds a strongly stable matching when one side has partial order preferences and the other side has complete total order preferences. A matching $M$ is strongly stable if it has not strong blocking pairs. $(w, f)$ is a strong blocking pair for $M$ under $X$ if $w$ strictly prefers $f$ to his current partner and $f$ either strictly prefers $w$ to her current partner of is indifferent between them. Given original instance $A$ and error instances $S = \{B_1, \dots B_k\}$, we can construct compound instance $X$ (Algorithm~\ref{alg:compoundinstance}) with  partial order preferences for firms and complete total order preferences for workers. 

\begin{lemma}
    A matching is stable under $A$ and all instances in $S$ if and only if it is a strongly stable matching under $X$.
\end{lemma}
\begin{proof}
Let $M$ be a matching.
\begin{claim}
A strong blocking pair $(w, f)$ of $M$ in $X$  is a blocking pair of $M$ in at least one of  $\{A\} \cup S$. 
\end{claim}
Let $(w, f)$ be a strong blocking pair of $M$ under $X$. Since $w$ doesn't change its preference from $A$ to $S$, $w$ must strictly prefer $f$ to $M(w)$. Then $f$ must either strictly prefer $w$ to its partner in $M$ or is indifferent between them. In both cases, $f$ must prefer $w$ to its partner $M(f)$ in at least one instance in $\{A\} \cup S$. Hence $(w, f)$ is a blocking pair for that instance.

\begin{claim}
A blocking pair of $M$ in one of $\{A\} \cup S$ is a strong blocking pair of $M$ in $X$. 
\end{claim}
Without loss of generality, assume $(w, f)$ is a blocking pair for $A$. Then $w$ and $f$ both prefer each other to their partners in $M$. In $X$, $w$ must still prefer $f$ to its partner since it does not change preferences. In addition, $f$ must either strictly prefer $w$ to its partner or be indifference between them. Hence $(w, f)$ is a strong blocking pair for $X$. 

 \end{proof}

\begin{figure}
	\begin{wbox}
		\textsc{CompoundInstance}($A, S = \{B_1, \dots B_k\}$): \\
		\textbf{Input:} Stable matching instance $A$, Set of instances with errors on firms side $S = \{B_1, \dots B_k\}$. \\
		\textbf{Output:} Instance $X$.
		\begin{enumerate}
			\item $\forall w \in W, w$'s preferences in $X$ and $A$ are the same.
			\item $\forall f \in F, \forall w_1, w_2 \in W, w_1 >_f^X w_2$ if and only if $w_1 >_f^A w_2$ and $\forall B_i \in S, w_1 >_f^{B_i} w_2$ 
			\item Return $X$.
		\end{enumerate}
	\end{wbox}
	\caption{Subroutine for constructing a compound instance.}
	\label{alg:compoundinstance} 
\end{figure} 

\begin{figure}[h]
	\begin{wbox}
		\textsc{AlgorithmForErrorsOnFirmsSide}($A, S = \{B_1, \dots B_k\}$): \\
		\textbf{Input:} Stable matching instance $A$, Set of instances with errors on firms side $S = \{B_1, \dots B_k\}$. \\
		\textbf{Output:} Stable matching $M$ or $\boxtimes$.
		\begin{enumerate}
			\item Construct instance $X$ = \textsc{Compound Instance($A,  S)$}
			
			\item Workers maintain a list of firms in accordance with their preference order in $X$.
			\item Firms maintain a set of all worker proposals received so far, initialized to $\emptyset$. 
			
			\item Until all firms receive a proposal or a worker can’t propose anymore, do
			\begin{itemize}
				\item $\forall w \in W, w$ proposes to its best uncrossed firm in its list.
				\item $\forall f \in F, f$ tentatively accepts the best proposal in its set and rejects the rest.
				\item $\forall w \in W,$ if $w$ is rejected by a firm $f$, cross $f$ off its list. 
				
			\end{itemize} 
			\item Return perfect matching $M$ or  $\boxtimes$.
		\end{enumerate}
	\end{wbox}
	\caption{Algorithm for finding a fully robust stable matching with errors only on firms' side.}
	\label{alg:daalgorithm} 
\end{figure} 

Every step in the algorithm can be executed in polynomial time (note that $|S|$ is polynomial). There are at most $n^2$ iterations as at least one worker should cross off a firm in each iteration. So, Algorithm~\ref{alg:daalgorithm} runs in polynomial time.

\subsection{Errors on both the Workers and Firms Sides}

\begin{figure}[h]
	\begin{wbox}
		\textsc{AlgorithmForErrorsOnWorkersAndFirmsSide}($A, S_1 = \{B_1, \dots B_k\}, S_2 = \{C_1, \dots C_k\}$): \\
		\textbf{Input:} Stable matching instance $A$, Set of instances with errors on firms side $S_1 = \{B_1, \dots B_k\}$ and errors on workers side $S_2 = \{C_1, \dots C_k\}$. \\
		\textbf{Output:} Stable matching $M$ or $\boxtimes$.
		\begin{enumerate}
			\item Construct instance $X$ = \textsc{Compound Instance($A,  S_1)$}
			
			\item Workers maintain a list of firms in accordance with their preference order in $A$.
			\item Firms maintain a set of all worker proposals received so far, initialized to $\emptyset$. 
			
			\item Until all firms receive a proposal or a worker can’t propose anymore, do
			\begin{enumerate}
				\item $\forall w \in W, w$ proposes to its best uncrossed firm in its list.
				\item $\forall f \in F, f$ tentatively accepts the best proposal in its set and rejects the rest.
				\item $\forall w \in W,$ if $w$ is rejected by a firm $f$, cross $f$ off its list. 
			\end{enumerate} 
			
			\item If a worker is rejected by all firms, return $\boxtimes$. Else every worker is matched. Let the perfect matching be $M$.
			
			\item Check if $M$ is stable under all instances $A, S_1$ and $S_2$. If so, return $M$. Else, let $(w', f)$ be a blocking pair for $M$ under one of the instances with $M(f) = w$ and $M(w') = f'$. Non-deterministically choose one of steps 7 and 8.
			\label{op6}
			
			\item Let $f$ reject $w$ and $w$ cross off $f$ from its list. GOTO step (4c).
			
			\item Let $f'$ reject $w'$ and $w'$ cross off $f'$ from its list. GOTO step (4c).
		\end{enumerate}
	\end{wbox}
	\caption{Non-deterministic algorithm for finding a fully robust stable matching with errors on workers and firms sides.}
	\label{alg:daalgorithm2} 
\end{figure}

Algorithm~\ref{alg:daalgorithm2}, an extension of Algorithm~\ref{alg:daalgorithm}, can find a stable matching when both workers and firms are allowed to change their preferences. The motivation for this algorithm is based on the lattice structure formed by the set of stable matchings. These lattices can be traversed using the rotations explained in section Section~\ref{sec.lattice}. For a given matching, a rotation involving a firm, if it exists, can be obtained by asking the firm to reject their partner and continuing the Deferred Acceptance algorithm from that matching. 

The algorithm first finds a stable matching under all instances in $\{A\} \cup S_1$. It checks if the matching is also stable under all of the instances in $S_2$. If not, it traverses the lattice of stable matching by looking at the blocking pairs. There are two ways in which a blocking pair can be removed - by ``breaking" the pairings of either of the involved parties. Hence for each blocking pair, there are two rotations that can lead to a stable matching. 

The instance given in Figure~\ref{ex:example1} shows that choosing Step 8 alone does not suffice. Similarly, Figure~\ref{ex:example2} shows that choosing Step 7 alone does not suffice. The output of the current algorithm depends on which rotation is performed; for it to be sure that no stable matching solution exists, it has to check all possible ways of breaking up all blocking pairs. The deterministic version of this algorithm will have \textbf{exponential} runtime - a larger example using $O(n)$ 'units' of example 1 and 2 each would need to check $O(2^n)$ paths. Finding a better way of removing blocking pairs poses an interesting open problem. 



\begin{figure}
	\begin{wbox}
	\begin{minipage}{.33\linewidth}
		\centering
		\begin{tabular}{l|lll}
			a&3&1&2 \\
            b&1&2&3 \\
            c&1&2&3 \\
		\end{tabular}
		
		\hspace{1cm}
		
		firm preferences in $A$
		
		\hspace{1cm}
		
	\end{minipage}%
	\begin{minipage}{.34\linewidth}
	\centering
	\begin{tabular}{l|lll} 
		1&a&c&b\\
        2&b&c&a\\
        3&c&a&b\\
	\end{tabular}
	
	\hspace{1cm}
	
	worker preferences in $A$
	 
	\hspace{1cm}
\end{minipage}%
	\begin{minipage}{.33\linewidth}
		\centering
		\begin{tabular}{l|lll}
		1&a&c&b\\
        2&c&b&a\\
        3&c&a&b\\
		\end{tabular}

		\hspace{1cm}

		worker preferences in $C_1$
		\hspace{1cm}
	\end{minipage} 
	\end{wbox}
	\caption{An example in which Step 8 does not suffice}
	\label{ex:example1}
	
	\begin{wbox}
	\begin{minipage}{.33\linewidth}
		\centering
		\begin{tabular}{l|lll}
			a&2&1&3 \\
            b&1&2&3 \\
            c&2&3&1 \\
		\end{tabular}
		
		\hspace{1cm}
		
		firm preferences in $A$
		
		\hspace{1cm}
		
	\end{minipage}%
	\begin{minipage}{.34\linewidth}
	\centering
	\begin{tabular}{l|lll}
		1&c&a&b\\
        2&b&a&c\\
        3&c&b&a\\
	\end{tabular}
	
	\hspace{1cm}
	
	worker preferences in $A$
	 
	\hspace{1cm}
\end{minipage}%
	\begin{minipage}{.33\linewidth}
		\centering
		\begin{tabular}{l|lll}
		1&c&a&b\\
        2&a&c&b\\
        3&c&b&a\\
		\end{tabular}

		\hspace{1cm}

		worker preferences in $C_1$
		\hspace{1cm}
	\end{minipage} 
	\end{wbox}
	\caption{An example in which Step 7 does not suffice}
	\label{ex:example2} 
\end{figure}



\end{document}